\documentclass{elsarticle}

\usepackage{graphicx,subfigure}
\usepackage{amsmath,amsthm,amssymb}
\usepackage{ulem}

\usepackage{xcolor}

\DeclareGraphicsExtensions{.pdf}
\usepackage{algorithm,algorithmic}

\def\P{\Pi \circ}
\newcommand{\vv}[1]{\boldsymbol{#1}} 
\newcommand{\mean}[2]{{\mu}_{{#1}{#2}}} 
\newcommand{\measmean}[2]{\tilde{\mu}_{{#1}{#2}} } 
\newcommand{\noise}[1]{{\eta}_{#1} }

\usepackage{xcolor}

\newtheorem{lemma}{Lemma}

\title{Low-Delay High-Rate Operation of 802.11ac WLAN Downlink: Nonlinear Controller Analysis \& Design}
\author[1]{Francesco Gringoli}
\address[1]{University of Brescia, Italy}
\author[2]{Douglas J. Leith\corref{cor1}}
\ead{doug.leith@tcd.ie}
\address[2]{Trinity College Dublin, Ireland}
\cortext[cor1]{Corresponding author doug.leith@tcd.ie}

\begin{document}

\global\csname @topnum\endcsname 0
\global\csname @botnum\endcsname 0

\begin{abstract}
In this paper we consider a next generation edge architecture where traffic is routed via a proxy located close to the network edge (e.g. within a cloudlet).  This creates freedom to implement new transport layer behaviour over the wireless path between proxy and clients.   We use this freedom to develop a novel traffic shaping controller for the downlink in 802.11ac WLANs that adjusts the send rate to each WLAN client so as to maintain a target number of packets aggregated in each transmitted frame.   In this way robust low-delay operation at high data rates becomes genuinely feasible across a wide range of network conditions.  Key to achieving robust operation is the design of an appropriate feedback controller, and it is this which is our focus.   We develop a novel nonlinear control design inspired by the solution to an associated proportional fair optimisation problem.   The controller compensates for system nonlinearities and so can be used for the full envelope of operation.   The robust stability of the closed-loop system is analysed and the selection of control design parameters discussed.     We develop an implementation of the nonlinear control design and use this to present a performance evaluation using both simulations and experimental measurements.
\end{abstract}
\maketitle

\section{Introduction}
In this paper we consider next generation edge transport architectures of the type illustrated in Figure \ref{fig:edge}(a) with the objective of achieving high rate, low latency communication on the downlink.   Traffic to and from client stations is routed via a proxy located close to the network edge (e.g. within a cloudlet).  This creates the freedom to implement new transport layer behaviour over the path between proxy and clients, which in particular includes the last wireless hop.   

For high rate, low latency communication we would like to select a downlink send rate which is as high as possible yet ensures that a persistent queue backlog does not develop at the wireless access point (AP).    Recently, it has been shown that the number of packets aggregated in transmitted frames is strongy coupled to the queue backlog at the AP and that adjusting the send rate so as to regulate the aggregation level acts to regulate the queueing delay at the AP~\cite{quickandplenty}.     

\begin{figure}
\centering
\subfigure[]{
\raisebox{0.25in}{\includegraphics[width=0.47\columnwidth]{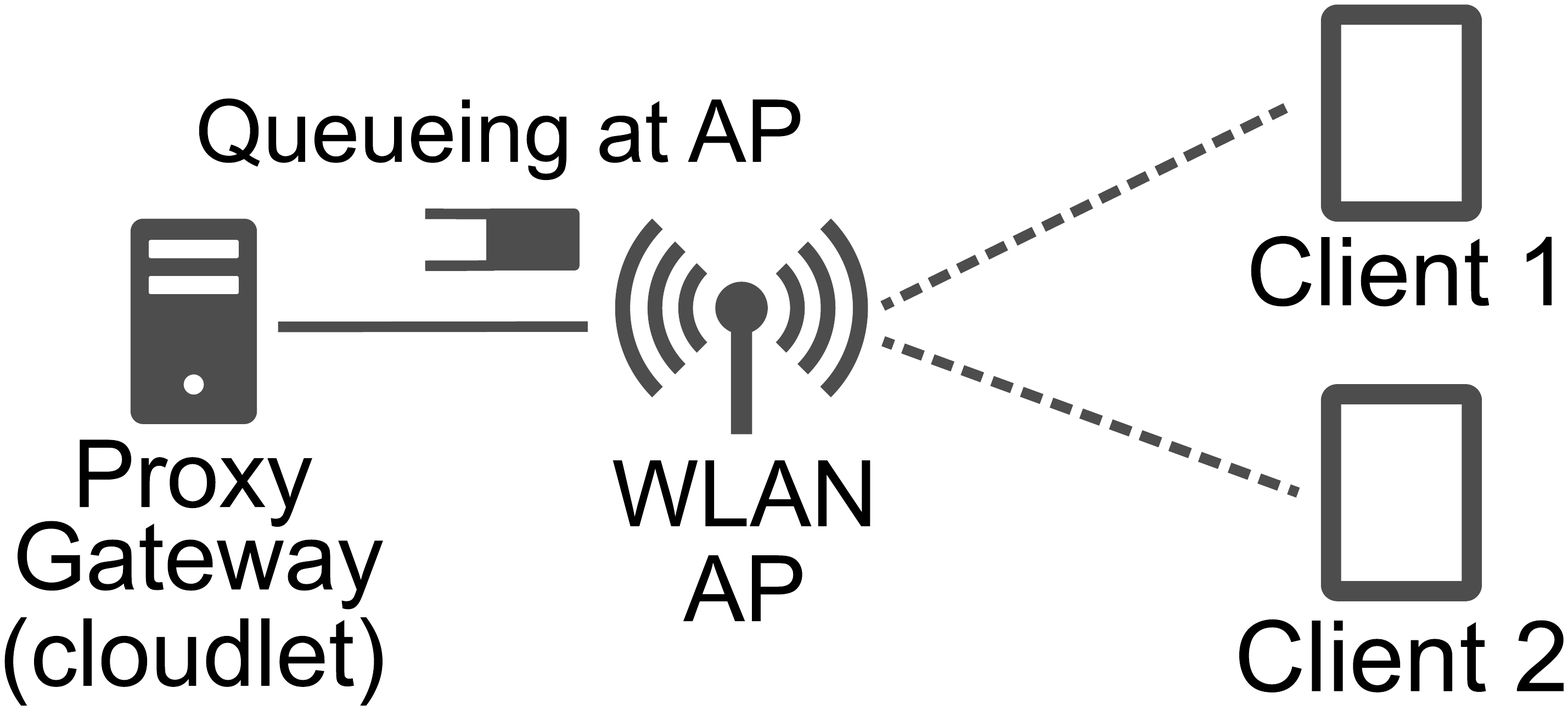}}
}
\subfigure[]{
\includegraphics[width=0.47\columnwidth]{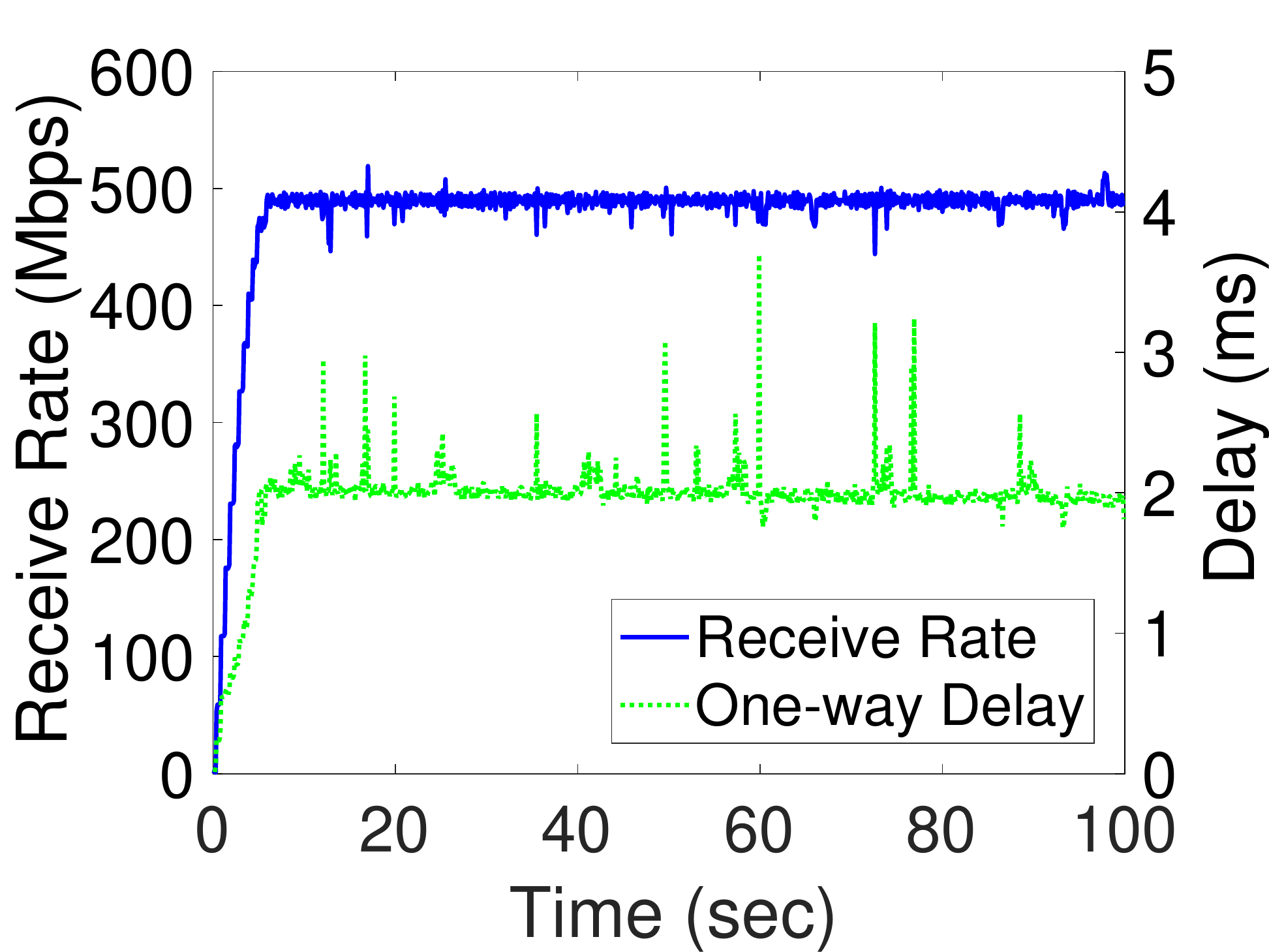}
}
\caption{(a) Cloudlet-based edge transport architecture with bottleneck in the WLAN hop (therefore queueing of downlink packets occurs at the AP as indicated on schematic). One great advantage of this architecture is its ease of rollout since the new transport can be implemented as an app on the clients plus a proxy deployed in the cloud; no changes are required to existing AP's or servers.  (b) Illustrating low-latency high-rate operation in an 802.11ac WLAN (measurements are from a hardware testbed located in an office environment).}
\label{fig:edge}
\end{figure}

However, the relationship between send rate, aggregation level and queueing delay is strongly nonlinear.  As a result it is challenging to design a stable, well-behaved feedback controller that dynamically adjusts send rate so as to regulate aggregation level and delay as network conditions change.   In this paper we develop a novel rate control algorithm that completely solves this task.    The controller compensates for system nonlinearities and so can be used for the full range of network operation.   The robust stability of the closed-loop system is analysed and the selection of control design parameters discussed.     We develop an implementation of the nonlinear controller and use this to present a performance evaluation using both simulations and experimental measurements.  

Figure \ref{fig:edge}(b) shows typical results obtained from a hardware testbed located in an office environment.  It can be seen that the one-way delay is low, at around 2ms, while the send rate is high, at around 500Mbps (this data is for an 802.11ac downlink using three spatial streams and MCS 9).  Increasing the send rate further leads to sustained queueing at the AP and an increase in delay, but the results in Figure \ref{fig:edge}(b) illustrate the practical feasibility of operation in the regime where the rate is maximised subject to the constraint that sustained queueing is avoided.  

A notable feature of our controller design is that it is closely related to a gradient descent algorithm for solving a reformulated proportional fair utility optimisation problem. It therefore establishes a new and interesting connection between feedback control and online optimisation.

\section{Related Work}

The closest related work is \cite{quickandplenty}, which introduced the idea of using aggregation to regulate send rate in WLANs.  They use a linear PI controller but no stability analysis is provided and consideration is confined to a small operational envelope where the system exhibits linear behaviour.    

Control theoretic analysis of WLANs has received relatively little attention in the literature, and has almost entirely focussed on MAC layer resource allocation, see for example \cite{BanchsSA06,BoggiaCGM07,Garcia-SaavedraBSW12,SerranoPMMB13} and references therein.    In contrast, there exists a substantial literature on control theoretic analysis of congestion control at the transport layer, see for example the seminal work in \cite{Kelly1998,JinWL04}.  However, this has mainly focussed on end-to-end behaviour in wired networks with queue overflow losses and has largely ignored detailed analysis of operation over WLANs.   This is perhaps unsurprising since low delay operation at the network edge has only recently come to the fore as a key priority for next generation networks.

TCP BBR~\cite{BBR} is currently being developed by Google and this also targets high rate and low latency, although not specifically in edge WLANs.  The BBR algorithm tries to estimate the bottleneck bandwidth and adapt the send rate accordingly to try to avoid queue buildup.  The delivery rate in BBR is defined as the ratio of the in-flight data when a packet departed the server to the elapsed time when its ACK is received.  This may be inappropriate, however, when the bottleneck is a WLAN hop since aggregation can mean that increases in rate need not correspond to increases in delay plus a small queue at the AP can be benificial for aggregation and so throughput.

\section{Preliminaries}
\subsection{Aggregation In 802.11ac}
A feature shared by all WLAN standards since 2009 (when 802.11n was introduced) has been the use of aggregation to amortise PHY and MAC framing overheads across multiple packets.   This is essential for achieving high throughputs.   {Since the PHY overheads are largely of fixed duration, increasing the data rate reduces the time spent transmitting the frame payload but leaves the PHY overhead unchanged.}  Hence, the efficiency, as measured by the ratio of the time spent transmitting user data to the time spent transmitting an 802.11 frame, decreases as the data rate increases unless the frame payload is also increased i.e. several data packets are aggregated and transmitted in a single frame. 

{In 802.11ac two types of aggregation are used.  Namely, one or more ethernet frames may aggregated into a single 802.11 AMSDU frame, and multiple AMSDUs then aggregated into an 802.11 AMPDU frame.   Typically, two ethernet frames are aggregated into one AMSDU and up to 64 AMSDUs aggregated into one AMPDU, allowing up to 128 ethernet frames to be sent in a single  AMPDU.   
}

The level of aggregation can be readily measured at a receiver either directly from packet MAC timestamps when these are available (which depends on hardware support and the availability of root access) or by applying machine-learning techniques to packet kernel timestamps (which are always available via the standard socket API), see \cite{quickandplenty} for details.

\subsection{Modelling Aggregation Level in Paced 802.11ac WLANs}
For control analysis and design we need a model of the relation between send rate and aggregation level on the downlink.  The model does not need to be exact since the feedback action of the controller can compensate for model uncertainty, but should capture the broad relationship.   In general, finite-load analytic modelling of aggregation in 802.11ac is challenging due to: (i) the randomness in the time between frame transmissions caused by the stochastic nature of the CSMA/CA MAC and (ii) correlated bursty packet arrivals.   As a result most finite-load analysis of packet aggregation to date has resorted to use of simulations.   Fortunately, since we control the sender we can use packet pacing and when packets are paced simple yet accurate analytic models are known~\cite{model}.   Following \cite{model}, let $n$ denote the number of client stations in the WLAN, $x_i$ the mean send rate to client $i$ in packets/sec and $\vv{x}=(x_1,\dots,x_n)$ the vector of mean send rates.   Let $\mean{T}{_{oh}}$ denote the time per frame used for CSMA/CA channel access, PHY and MAC headers plus transmission of the MAC ACK so that $c=n\mean{T}{_{oh}}$ is the aggregate overhead for a round of transmissions to the $n$ client stations, $w_i$ the mean time to transmit one packet to client $i$ at the PHY data rate and $\vv{w}=(w_1,\dots,w_n)^T$ the vector of transmit times, $N_{max}$ the maximum number of packets that can be sent in a frame (typically 32 or 64).   Parameters $c$ and $\vv{w}$ capture the WLAN PHY/MAC configuration.   Define $\mean{N}{_i}$ to be the mean number of packets in each frame sent to client $i$ and $\mean{\vv{N}}{}=(\mean{N}{_1},\dots,\mean{N}{_n})$ the vector of mean aggregation levels.  Then,
\begin{align}
\mean{\vv{N}}{}=\Pi\frac{c\vv{x}}{1-\vv{w}^T\vv{x}} =\Pi \vv{F}(\vv{x})\label{eq:model}
\end{align}
where $\Pi$ denotes projection onto interval $[1,N_{max}]$ and $\vv{F}(\vv{x}):=\frac{c\vv{x}}{1-\vv{w}^T\vv{x}}$.   The mean delay $\mean{T}{_i}$ of packets sent to client $i$ is upper bounded by
\begin{align}
\mean{T}{_i} &=\max\{\min\{\frac{c}{1-\vv{w}^T\vv{x}},\frac{N_{max}}{x_i}\},\frac{1}{x_i}\}\label{eq:delay}
\end{align}

Observe that $\vv{F}(\vv{x})$ is monotonically increasing for feasible rate vectors $\vv{x}$ since $\frac{\partial{F}_i^\prime(\vv{x})}{\partial x_i}=\frac{c}{(1-\vv{w}^T\vv{x})^2}(1-\vv{w}^T\vv{x}+w_ix_i)>0$ and $\frac{\partial{F}_i^\prime(\vv{x})}{\partial x_j}=\frac{cw_ix_i}{(1-\vv{w}^T\vv{x})^2}>0$ when $x_i\ge 0$ and $\vv{w}^T\vv{x}<1$.   Hence, $\vv{F}(\vv{x})$ is one-to-one and so invertible.  In particular, 
\begin{align}
\vv{F}^{-1}({\mean{\vv{N}}{}})=\frac{{\mean{\vv{N}}{}}}{c+\vv{w}^T{\mean{\vv{N}}{}}}\label{eq:inv}
\end{align}
and it can be verified that $\vv{F}(\vv{F}^{-1}({\mean{\vv{N}}{}}))={\mean{\vv{N}}{}}$.   
Given rate vector $\vv{x}$ we can therefore obtain the corresponding aggregation level from $F(\vv{x})$ and, conversely, given aggregation level vector ${\mean{\vv{N}}{}}$ we can obtain the corresponding rate vector from $\vv{F}^{-1}(\mean{\vv{N}}{})$.   This will prove convenient in the analysis below since it means we can freely change variables between $\vv{x}$ and ${\mean{\vv{N}}{}}$.   For example, substituting $\vv{x}=\vv{F}^{-1}(\mean{\vv{N}}{})$ the term $\frac{c}{1-\vv{w}^T\vv{x}}$ in the mean delay (\ref{eq:delay}) can be expressed equivalently in terms of ${\mean{\vv{N}}{}}$ as,
\begin{align}
\frac{c}{1-\vv{w}^T\vv{x}}= c+\vv{w}^T\mean{\vv{N}}{}\label{eq:handyT}
\end{align}
\subsection{Model Uncertainty}
A key consideration is that the model (\ref{eq:model}) is only approximate and so simply applying a utility fair rate allocation derived from this model may lead to poor behaviour (e.g. unfairness, higher delay, lower rate) in the actual WLAN.  Instead we require a feedback-based approach that measures the actual WLAN behaviour and adjusts the send rates so as to achieve fair high-rate, low-delay behaviour.  

With this in mind we would also like to characterise the main sources of modelling error.  In model (\ref{eq:model}) the function $\vv{{F}}$ relating mean aggregation level to send rate involves send rates $\vv{x}$, PHY MCS rates $\vv{w}$ and MAC overhead parameter $c=n\mean{T}{_{oh}}$.  The send rates are known since we control the sender, and the PHY MCS rates of received frames can be readily observed at client stations and reported back to the sender.  However, parameter $c$ cannot be easily measured and is only approximately known.  This is because the mean channel access time $\mean{T}{_{oh}}$ cannot be measured directly (since we consider the transport layer we assume we do not have access to the MAC on the AP) and it depends on the channel state and so may be strongly affected by neighbouring WLANS, interference etc.   Hence, only a fairly rough estimate of parameter $c$ is generally available and this is the main source of parameter uncertainty in the model.  

To distinguish between the model and the actual WLAN we will use $\vv{{F}}$ to refer to the model and $\vv{\tilde{F}}$ to refer to the actual WLAN behaviour.   In more detail, we divide time into slots of duration $\Delta$ and suppose that the send rate is held constant over a slot.   Let $\Phi_i(k)$ be the set of frames received by client station $i$ in slot $k$.  The client $i$ can observe the number of packets ${N}_{i,f}$ in each received frame $f\in\Phi_i(k)$ and calculate the empirical mean aggregation level
\begin{align}
\measmean{{N}}{_i}(k)=\frac{1}{|\Phi_i(k)|}\sum_{f\in \Phi_i(k)}{N}_{i,f}
\end{align} 
Letting $\vv{\tilde{F}}_k=E[\frac{1}{|\Phi_i(k)|}\sum_{f\in \Phi_i(k)}{N}_{i,f}]$ and $\vv{x}_k$ be the send rate in slot $k$ then
\begin{align}
\measmean{{N}}{_i}(k)=\vv{\tilde{F}}_k(\vv{x}_k)+\noise{\measmean{\vv{N}}{}}\label{eq:meas}
\end{align} 
where $\noise{\measmean{\vv{N}}{}}$ is the measurement noise. Due to mismatches between the model and the real system, in general $\vv{\tilde{F}}_k(\vv{x}_k)\ne \vv{F}(\vv{x}_k)$.    

%
%

\section{Proportional Fair Low-Delay Rate Allocation}

\subsection{Utility Fair Optimisation Is Convex }
Our interest is in achieving high rates while maintaining low delay at the AP.   We begin by first considering the proportional fair low delay rate allocation, which is the solution to the following utlity-fair optimisation $P$: 
\begin{align}
&\max_{\vv{x}\in\mathbb{R}^n_+}\sum_{i=1}^n \log x_i  \label{eq:obj}\\
s.t.\ &  \mean{T}{_i}(\vv{x})\le \bar{T}, i=1,\dots,n \label{eq:con1}\\
&\mean{N}{_i}(\vv{x})\le \bar{N}, i=1,\dots,n \label{eq:con2}
\end{align}
Constraint (\ref{eq:con1}) ensures that the mean delay at the AP is no more than upper limit $\bar{T}$, where $\bar{T}$ is a QoS parameter.  Constraint (\ref{eq:con2}) ensures that we operate at an aggregation level no more than $\bar{N}<N_{max}$.  

Substituting from (\ref{eq:delay}) the constraints (\ref{eq:con1}) can be written\footnote{Note that constraint (\ref{eq:con2}) ensures  $\mean{N}{_i}(\vv{x})\le \bar{N}<N_{max}$ and so $ \mean{T}{_i}(\vv{x})<N_{max}/x_i$.  Since our interest is primarily in applications requiring high rates we assume for simplicity that $\frac{c}{1-\vv{w}^T\vv{x}}\ge\frac{1}{x_i}$ although this could be added as the additional linear constraint $cx_i +  \vv{w}^T\vv{x} \ge  1$ if desired.} as $\frac{c}{1-\vv{w}^T\vv{x}}\le \bar{T}$.  Rearranging gives $c\le \bar{T}(1-\vv{w}^T\vv{x})$ i.e. $\vv{w}^T\vv{x} \le 1-c/\bar{T}$.   In this form it can be seen that the constraint is linear, and so convex.   Similarly, substituting from (\ref{eq:model}) the constraints (\ref{eq:con2}) can be written equivalently as $\frac{cx_i}{1-\vv{w}^T\vv{x}} \le  \bar{N}$, $i=1,\dots,n$.  Rearranging gives $cx_i\le  \bar{N}(1-\vv{w}^T\vv{x})$ i.e. $cx_i +  \bar{N}\vv{w}^T\vv{x} \le  \bar{N}$, which again is linear.   Hence, optimisation $P$ can be equivalently rewritten as optimisation $P^\prime$:
\begin{align}
&\max_{\vv{x}\in\mathbb{R}^n_+}\sum_{i=1}^n \log x_i\\
s.t.\ & \vv{w}^T\vv{x} \le 1-c/\bar{T} \label{eq:con1b}\\
&cx_i +  \bar{N}\vv{w}^T\vv{x} \le  \bar{N}, i=1,\dots,n \label{eq:con2b}
\end{align}
which is convex. 

\subsection{Characterising The Proportional Fair Solution}\label{sec:soln}
The Lagrangian of optimisation $P^\prime$ is $L(\vv{x},\theta,\vv{\lambda}):=-\sum_{i=1}^n \log x_i +\theta( \vv{w}^T\vv{x}-(1-c/\bar{T}))+ \sum_{i=1}^n\lambda_i(cx_i +  \bar{N}\vv{w}^T\vv{x} - \bar{N})$ where $\theta$ and $\lambda_i$, $i=1,\dots,n$ are multipliers associated with, respectively, (\ref{eq:con1b}) and (\ref{eq:con2b}).  Since the optimisation is convex the KKT conditions are necessary and sufficient for optimality.  Namely, an optimal rate vector $\vv{x}^*$ satisfies
\begin{align}
-\frac{1}{x_i^*} +\lambda_i c +\sum_{j=1}^n\lambda_j  \bar{N}w_i+\theta w_i=0
\end{align}
i.e.
\begin{align}
x_i^*=\frac{1}{\lambda_i c +Dw_i}
\end{align}
where $D:=( \bar{N}\sum_{j=1}^n\lambda_j +\theta )$.

Let $U=\{i:\mean{N}{_i}(\vv{x}^*)< \bar{N}\}$ denote the set of stations for which the aggregation level is strictly less than $ \bar{N}$ at the optimal rate allocation.  By complementary slackness $\lambda_i=0$ for $i\in U$ and so $x_i^*=1/(Dw_i)$.   That is, $\mean{N}{_i}=\frac{cx_i^*}{(1-\vv{w}^T\vv{x}^*)}=\frac{c}{D(1-\vv{w}^T\vv{x}^*)}\frac{1}{w_i}$.   Observe that the first term is invariant with $i$ and so the aggregation level of station $i\in U$ is proportional to $1/w_i=\mean{R}{_i}/L$ i.e. to the mean MCS rate of the station.  For stations $j\notin U$ the aggregation level $\mean{N}{_j}(\vv{x}^*)= \bar{N}$.  

Putting these observations together, it follows that
\begin{align}
\mean{N}{_i}(\vv{x}^*)=\min\{\frac{c}{D(1-\vv{w}^T\vv{x}^*)}\frac{1}{w_i},  \bar{N}\},\ i=1,\dots,n
\end{align}
Assume without loss that the station indices are sorted such that $w_1\ge w_2\ge \dots \ge w_n$.  Then 
\begin{align}
\mean{N}{_i}(\vv{x}^*)=\min\{\mean{N}{_{1}}(\vv{x}^*)\frac{w_{1}}{w_i},  \bar{N}\},\ i=2,\dots,n\label{eq:soln1}
\end{align}
Hence, once the optimal $\mean{N}{_1}(\vv{x}^*)$ is determined we can find the optimal aggregation levels for the rest of the stations.  With these we can then use inverse mapping (\ref{eq:inv}) to recover the proportional fair rate allocation, namely $x_i^*=\mean{N}{_i}/(c+\vv{w}^T\mean{\vv{N}}{})$.   

It remains to determine $\mean{N}{_1}$.  We proceed as follows.
\begin{lemma}\label{lem:one}
At an optimum $\vv{x}^*$ of $P^\prime$ then either (i) $\mean{N}{_i}(\vv{x}^*)= \bar{N}$ for all $i=1,\dots,n$ or (ii) $\mean{T}{_i}(\vv{x}^*)=\bar{T}$ for all $i=1,\dots,n$.
\begin{proof}
We proceed by contradiction.  Suppose at an optimum $\mean{N}{_i}(\vv{x}^*)=\frac{cx_i^*}{1-\vv{w}^T\vv{x}^*}< \bar{N}$ for some $i$ and $\mean{T}{_i}(\vv{x}^*)=\frac{c}{1-\vv{w}^T\vv{x}^*}<\bar{T}$.  Then we can increase $x^*_i$ without violating the constraints (with this change $\frac{c}{1-\vv{w}^T\vv{x}^*}$ and $\frac{cx_i^*}{1-\vv{w}^T\vv{x}^*}$ will both increase, but since the corresponding constraints are slack if the increase in $x^*_i$  is sufficiently small then they will not be violated).  Hence, we can improve the objective which yields the desired contradiction since we assumed optimality of $\vv{x}^*$.  Hence when $\mean{N}{_i}(\vv{x}^*)< \bar{N}$ for at least one station then $\mean{T}{_i}(\vv{x}^*)=\bar{T}$.  Alternatively, $\mean{N}{_i}(\vv{x}^*)= \bar{N}$ for all stations.
\end{proof}
\end{lemma}
It follows from Lemma \ref{lem:one} that $\mean{N}{_1}=\min\{\bar{T}x_1^*,  \bar{N}\}$.   Substituting into (\ref{eq:soln1}) and combining with inverse mapping (\ref{eq:inv}) it follows that
\begin{align}
\vv{x}^*&=\vv{F}^{-1}(\mean{\vv{N}}{}(\vv{x}^*))\label{eq:soln2a}
\end{align}
with
\begin{align}
\mean{\vv{N}}{}(\vv{x}^*)&=\min\{\bar{T}x_1^*,  \bar{N}\}\vv{W}\label{eq:soln2b}
\end{align}
where vector $\vv{W}=[1,\frac{w_{1}}{w_2},\dots,\frac{w_{1}}{w_n}]^T$. Recall that the station indices are sorted such that $w_1\ge w_2\ge \dots \ge w_n$ and so all of the elements of $\vv{W}$ are less than or equal to one, with equality only when the PHY data rate is maximal amongst the client stations.  The proportional fair rate $\vv{x}^*$ and associated aggregation level $\mean{\vv{N}}{}$ can now be found by solving equations (\ref{eq:soln2a})-(\ref{eq:soln2b}).   When $\bar{N}=+\infty$ or $\bar{T}=+\infty$ it can be verified that the proportional fair solution is an equal airtime one.  

\subsection{Offline Algorithm For Solving Optimisation $P^\prime$}
While we can solve convex optimisation $P^\prime$ using any standard algorithm, e.g. a primal-dual update, the following update will prove convenient since it lends itself readily to conversion to online form, as will be shortly discussed.  Letting $\vv{z}=\vv{F}(\vv{x})$ we can rewrite (\ref{eq:soln2a}) equivalently\footnote{Since $\mean{N}{_i}(\vv{x}^*)\le \bar{N}$ the minimiser satisfies $\vv{z}^*=\mean{N}{}(\vv{x}^*)$ i.e. $\vv{x}^*=\vv{F}^{-1}(\vv{z}^*)=\vv{F}^{-1}(\mean{\vv{N}}{}(\vv{x}^*))$ as required.} as
\begin{align}
\vv{z}^*\in\arg\min_{\vv{z}\in\mathbb{R}_+^n} \|\min\{\mean{\vv{N}}{}(\vv{x}^*),\bar{N}\}-\vv{z}\|_2^2
\end{align}
where the min of $\mean{\vv{N}}{}(\vv{x}^*)$ and $\bar{N}$ is applied elementwise.  Similarly, we can rewrite (\ref{eq:soln2b}) equivalently\footnote{The minimiser of the optimisation is $\nu=\min\{\bar{T}x_{1}^*, \bar{N}\}$ and so the solution is $\mean{N}{}(\vv{x}^*)=\min\{\bar{T}x_{1}^*, \bar{N}\}\vv{W}$, which matches (\ref{eq:soln2b}).} as 
\begin{align}
\nu(\vv{x}^*)&\in\arg\min_{\nu\in\mathbb{R}^+} (\nu-\min\{\bar{T}x_{1}^*, \bar{N}\})^2\\
\mean{\vv{N}}{}(\vv{x}^*)&=\nu(\vv{x}^*)\vv{W}
\end{align}
with $\vv{x}^*=\vv{F}^{-1}(\vv{z}^*)$.  Applying gradient descent to these two coupled optimisations yields update
\begin{align}
\vv{z}(k+1) &= \vv{z}(k) +K_1 (\min\{\nu(k)\vv{W},\bar{N}\}-\vv{z}(k))\label{eq:alg1}\\
\nu(k+1)&=[\nu(k)-K_2(\nu(k)-\min\{\bar{T}x_{1}(k), \bar{N}\})]^+\label{eq:alg2}\\
\vv{x}(k)&=\vv{F}^{-1}(\vv{z}(k))\label{eq:alg3}
\end{align}
with step sizes $K_1$, $K_2$, $[x]^+=x$ when $x\ge1$ and 1 otherwise, and where time is slotted with $\vv{z}(k)$ etc denoting the value in slot $k$.   Standard Lyapunov arguments (namely try a quadratic candidate Lyapunov function) can be used to show that (\ref{eq:alg1})-(\ref{eq:alg2}) converges to the proportional fair low-delay rate allocation provided step sizes $K_1, K_2$ are selected sufficiently small.

\section{Online Feedback Algorithm For Proportional Fair Low-Delay Rate Allocation}
The offline rate update (\ref{eq:alg1})-(\ref{eq:alg3}) lends itself directly to online feedback-based implementation by appropriately substituting the {measured} aggregation level $\measmean{\vv{N}}{}(k)=(\measmean{{N}}{_1}(k),\dots,\measmean{{N}}{_n}(k))$ observed in received frames over slot $k$, where recall from (\ref{eq:meas}) that $\measmean{{N}}{_i}(k)$ is the empirical mean number of packets per frame in frames received by client station $i$ in time slot $k$.   In the rest of this section we consider the resulting feedback controller in more detail.

\subsection{Inner Loop PI-Controller}
Replacing $\vv{z}(k)$ on the RHS of offline update (\ref{eq:alg1}) with the vector $\measmean{\vv{N}}{}(k)$ of measured WLAN aggregation levels gives online update
\begin{align}
\vv{z}(k+1) &= \vv{z}(k) +K_1 \vv{e}(k)\label{eq:inner1}
\end{align}
where vector $\vv{e}(k)=\vv{N}_{target}(k)-\measmean{\vv{N}}{}(k)$ is the ``error'' between the vector of desired aggregation levels $\vv{N}_{target}(k)$ and the actual aggregation levels $\measmean{\vv{N}}{}(k)$, with $\vv{N}_{target}(k)= \min\{\nu(k)\vv{W},\bar{N}\}$.   Importantly, observe that update (\ref{eq:inner1}) is an integral controller that adjusts $\vv{z}$ to try to regulate error $\vv{e}$ at zero.  Namely, when element $e_i$ of vector $\vv{e}$ is $>0$ then $z_i$ is increased which in turn tends to increase the aggregation level $\measmean{\vv{N}}{_i}$ of transmissions to station $i$ and so decrease error $e_i$, bringing it back towards zero.  Conversely, when $e_i<0$ then $z_i$ is decreased which tends to decrease the aggregation level and once again bring $e_i$ back towards zero.    The feedback from the measured aggregation level therefore tends to compensate for deviations in the WLAN from desired behaviour, intuitively making the WLAN behaviour robust to model uncertainty (we analyse robustness more formally below).

\begin{figure}
\centering
\includegraphics[width=0.8\columnwidth]{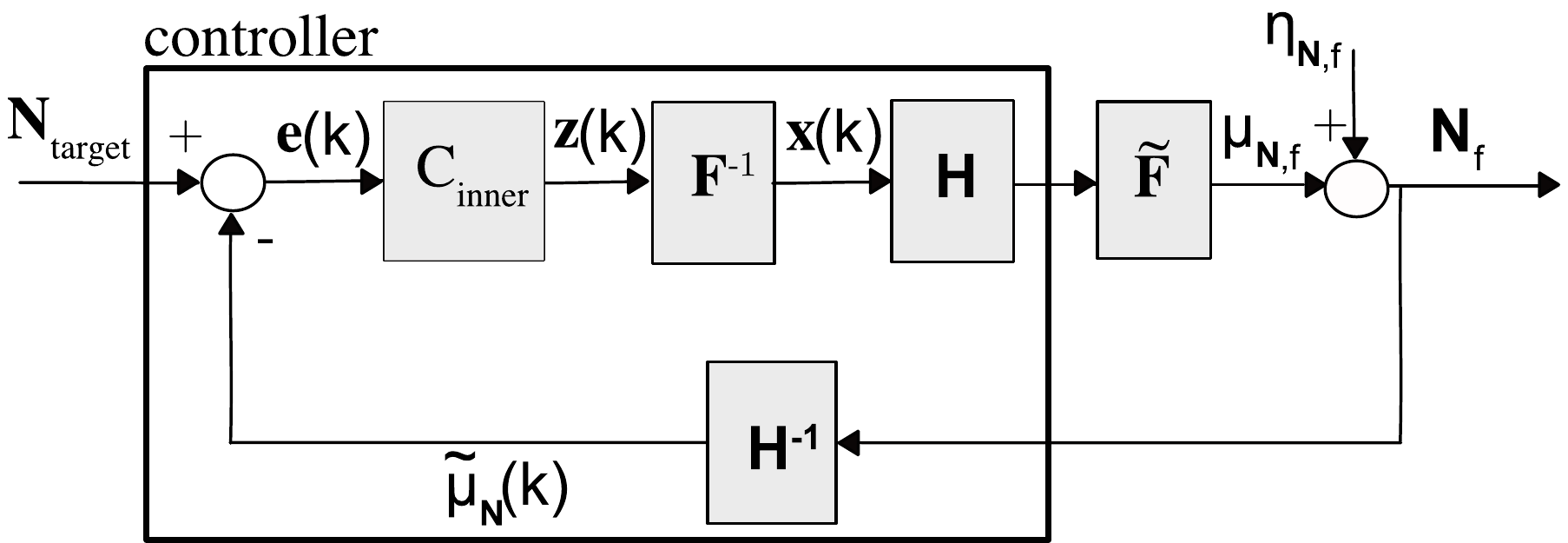}
\caption{Schematic of inner feedback loop corresponding to update ({eq:inner1}).  Controller updates occur at the start of time slots.  Nonlinear function $\vv{F}^{-1}$ is given by model (\ref{eq:inv}), $\vv{H}$ is a zero-order hold (i.e. holds the packet send rate constant at $\vv{x}(k)$ for packets sent during slot $[0,\Delta k)$), $\mean{\vv{N}}{_f}$ is the WLAN mean frame aggregation level, $\noise{\vv{N}_f}$ is measurement noise, $\vv{H}^{-1}$ maps from the individual frame aggregation levels to the empirical mean $\measmean{\vv{N}}{}(k)$ of the frame aggregation level over a slot.  }\label{fig:inner}
\end{figure}

\textit{Nonlinear Mixed Discrete/Continuous-Time Feedback Loop:} Figure \ref{fig:inner} shows schematically the feedback loop corresponding to (\ref{eq:inner1}).   Since the controller runs in discrete-time and the WLAN in continuous-time we need to map between the two using operators $\vv{H}$ and $\vv{H}^{-1}$.   $\vv{H}$ holds the sender inter-packet time equal to $1/\vv{x}(k)$ during slot $k$.  $\vv{H}^{-1}$ maps from the observed frame aggregation levels ovcer slot $k$ to the empirical average aggregation level $\measmean{\vv{N}}{}(k)$.   Nonlinear function $\vv{\tilde{F}}$ is the actual WLAN mapping from send rate to mean aggregation level and $\vv{\eta}_{\vv{N},f}$ the measurement noise induced by the use of the emprical mean $\measmean{\vv{N}}{}(k)$ rather than the actual mean.  $C_{inner}$ denotes the dynamic mapping (\ref{eq:inner1}) from $\vv{e}$ to $\vv{z}$.

\textit{Linearising Action of Inner-Loop Controller:}
Due to the nonlinear functions $\vv{F}^{-1}$ and $\vv{\tilde{F}}_k$ shown in Figure \ref{fig:inner} the update (\ref{eq:inner1}) is nonlinear in general.    However, when the model is exact then $\mean{\vv{N}}{_f}=\vv{\tilde{F}}_k(\vv{F}^{-1}(\vv{z}(k)))=\vv{z}(k)$ and the resulting update $\vv{z}(k+1)=\vv{z}(k)+K_1(\vv{N}_{target}-\vv{z}(k))$ is linear.  That is, the controller transforms the nonlinear system to have first-order linear dynamics.   More generally, we hope that our model is approximately correct so that $\vv{F}\approx\vv{\tilde{F}}_k$ and the inverse function $\vv{F}^{-1}$ in the controller tends to compensate for the nonlinearity $\vv{\tilde{F}}_k$ in the mapping from send rate to aggregation level.


\subsection{Outer Loop Controller}
Rewrite update (\ref{eq:alg2}) as
\begin{align}
\nu(k+1)&=[\nu(k)+K_2e_2(k)]^+\label{eq:outer1}
\end{align}
with $e_2(k)=\min\{\bar{T}x_{1}(k), \bar{N}\}-\nu(k)$ and $x_{1}(k)$ equal to the first element of vector $\vv{F}^{-1}(\vv{z}(k))$.   Observe that (\ref{eq:outer1}) again takes the form of an integral controller which tries to regulate $e_2(k)$ at zero i.e. to make $\nu(k)$ track $\min\{\bar{T}x_{1}(k), \bar{N}\}$.   

The input to the inner-loop controller is $\vv{N}_{target}(k)= \min\{\nu(k)\vv{W},\bar{N}\}$ and the combination of the inner-loop controller (\ref{eq:inner1}) with outer-loop controller (\ref{eq:outer1}) gives the  setup shown schematically in Figure \ref{fig:outer}.   It can be seen that we ``bootstrap'' from the inner loop and use $\min\{\bar{T}x_{1}(k), \bar{N}\}$ as the set point for outer loop control variable $\nu(k)$.  We then map from $\nu(k)$ to the target aggregation level $N_{target}$ using $\nu(k)\vv{W}$.     


\begin{figure}
\centering
\includegraphics[width=\columnwidth]{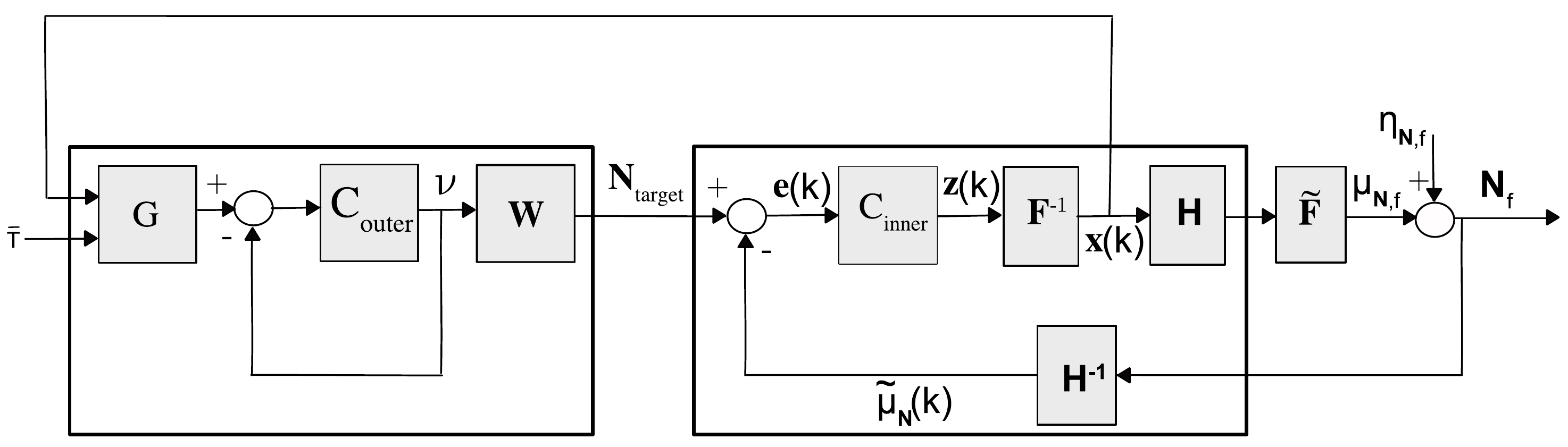}
\caption{Schematic of coupled feedback loops.   $\bar{T}$ is the target delay, $G(\vv{x}(k))=\min\{\bar{T}x_{1}(k), \bar{N}\}$, $C_{outer}$ denotes the outer control update (\ref{eq:outer1}).  Other quantities are as in Figure \ref{fig:inner}. 
}\label{fig:outer}
\end{figure}

\subsection{Closed-Loop Stability}
\subsubsection{Inner-Loop}
Recall that the main source of model uncertainty is parameter $c$.  That is, $\vv{F}(\vv{x})=\frac{c\vv{x}}{1-\vv{w}^T\vv{x}}$ whereas to a good approximation $\vv{\tilde{F}}_k(\vv{x})=\P\frac{\tilde{c}(k)\vv{x}}{1-\vv{w}^T\vv{x}}$ with $\tilde{c}_k\ne c$ (recall projection $\Pi$ captures the saturation constraint that $\measmean{\vv{N}}{}(k)\in[1,N_{max}]$).   Hence, $\vv{\tilde{F}}_k(\vv{F}^{-1}(\vv{z}(k)))=\P(\frac{\tilde{c}(k)}{c}\vv{z}(k))$ and dynamics (\ref{eq:inner1}) are  
\begin{align}
\vv{z}(k+1)&=\vv{z}(k)+K_1(\vv{N}_{target}-\Gamma(k)\vv{z}(k)) \label{eq:loop2}
\end{align}
where $\Gamma(k)=diag\{\gamma_1(k),\dots,\gamma_n(k)\}$ and $\gamma_i(k)=\frac{\P(\tilde{c}(k)z_i(k)/c)}{z_i(k)}$.

Neglecting the input $\vv{N}_{target}$ for the moment, it is easy to see\footnote{Try candidate Lyapunov function $V(k)=\vv{z}^T(k)\vv{z}(k)$.  Then $V(k+1)=(I-\Gamma(k))^T(I-\Gamma(k))V(k)$ (since $\Gamma(k)$ is diagonal) and so is strictly decreasing when $0< \gamma_i(k)<2$.} that the dynamics $\vv{z}(k+1)=(I-\Gamma(k))\vv{z}(k)$ are exponentially stable provided $0< \gamma_i(k)<2$ for all $i=1,\dots,n$.   Note that this stability holds for arbitrary time-variations in the $\gamma_i(k)$.  Projection $\Pi$ satisfies $0\le \frac{\P z}{z}\le1$ and $\tilde{c}(k),c$ are both non-negative, so for stability it is sufficient that $\tilde{c}(k)/c<2$.   This condition is also necessary since for constant $\tilde{c}(k)$ the  system will be unstable if this condition is violated.   

In summary, time-variations in the $\gamma_i(k)$ affect stability in a benign fashion and control parameter $c$ can safely be larger than the (uncertain) plant gain $\tilde{c}(k)$ (as this reduces the loop gain) but should not be too much smaller (since this increases the loop gain).  


\subsubsection{Outer-Loop}
The overall closed-loop system dynamics (\ref{eq:outer1})-(\ref{eq:loop2}) can be rewritten equivalently as,
\begin{align}
\nu(k+1)&=[\nu(k)-K_2(\nu(k)-G(\vv{F}^{-1}(\vv{z}(k))))]^+ \label{eq:outer3}\\
\vv{z}(k+1) &=\vv{z}(k)+K_1(\min\{\nu(k)\vv{W},\bar{N}\}-\vv{\Gamma}_k\vv{z}(k))  \label{eq:outer4}
\end{align}
where $G(\vv{x}(k))=\min\{\bar{T}x_{1}(k), \bar{N}\}$.  Assume the dynamics of the inner $\vv{z}$ loop are much faster than those of the outer $\nu$ loop (e.g. by selecting $K_2\ll K_1$) so that $\vv{z}(k)=\nu(k)\vv{W}$.   Then the system dynamics simplify to

\small
\begin{align}
&\nu(k+1)\notag\\
&=[\nu(k)-K_2(\nu(k)-\min\{\bar{T}\frac{\nu(k)}{c+\nu(k)nw_1},\bar{N}\})]^+\\
&=[\nu(k)-K_2(\nu(k)-\gamma_0(k)\bar{T}\frac{\nu(k)}{c+\nu(k)nw_1})]^+\\
&=[(1-K_2\frac{c+\nu(k)nw_1-\gamma_0(k)\bar{T}}{c+\nu(k)nw_1})\nu(k)]^+ \label{eq:insight}
\end{align}
\normalsize
where $0<\gamma_0(k)\le 1$ captures the impact of the $\bar{N}$ constraint i.e $\gamma_0(k)$ equals $1$ when $\bar{T}\frac{\nu(k)}{c+\nu(k)nw_1}\le \bar{N}$ and decreases as $\bar{T}\frac{\nu(k)}{c+\nu(k)nw_1}$ increases above $\bar{N}$.  We have also used the fact that $\vv{w}^T\vv{W}=nw_1$.   

We can gain useful insight into the behaviour of the system dynamics from inspection of (\ref{eq:insight}).   Namely, ignoring the constraints for the moment (i.e. $\gamma_0(k)=1$ and $\nu(k)\ge 1$) and assuming that $0<K_2<1$ then it can be seen that when $c+\nu(k)nw_1-\bar{T}<0$ then $1-K_2\frac{c+\nu(k)nw_1-\bar{T}}{c+\nu(k)nw_1})>1$ and so $\nu(k+1)$ increases (since $\nu(k)\ge 1$).  Hence $c+\nu(k)nw_1-\bar{T}$ increases until it equals 0 or becomes positive.  Conversely, when $c+\nu(k)nw_1-\bar{T}>0$ then $1-K_2\frac{c+\nu(k)nw_1-\bar{T}}{c+\nu(k)nw_1})<1$ and $\nu(k+1)$ decreases.  Hence, $c+\nu(k)nw_1-\bar{T}$ decreases until it equals 0 to becomes negative.  That is, the dynamics (\ref{eq:insight}) force $c+\nu(k)nw_1-\bar{T}$ to either converge to 0 or oscillate about 0.   

With the above in mind the impact of the constraints is now easy to see.   When $\bar{T}>c+\bar{N}nw_1$ then the delay target is hit at an aggregation level above $\bar{N}$.  It can be seen that $c+\nu(k)nw_1-\bar{T}<0$ for all admissible $\nu(k)$ and so $\nu(k)$ increases until it equals $\bar{N}$.   When $\bar{T}<c+nw_1$ then the target delay is violated even when the aggregation level is the minimum possible $\nu(k)=1$.  It can be seen that $c+\nu(k)nw_1-\bar{T}>0$ for all admissible $\nu(k)$ and so $\nu(k)$ decreases until it equals $1$.

To establish stability we need to show that persistent oscillations about $c+\nu(k)nw_1-\bar{T}=0$ cannot happen.  We have the following lemma:
\begin{lemma}\label{lem:two}
Suppose gain $0<K_2<1$ and initial condition $1\le \nu(1)\le\bar{N}$.  Then for the dynamics (\ref{eq:insight}) we have: (i) when $c+nw_1<\bar{T}<c+\bar{N}nw_1$ then $\nu(k)$ converges to $(\bar{T}-c)/(nw_1)$, (ii) when $\bar{T}\ge c+\bar{N}nw_1$ then $\nu(k)$ converges to upper limit $\bar{N}$ and (iii) when $\bar{T}<c+nw_1$ then $\nu(k)$ converges to lower limit $1$.
\end{lemma}
\begin{proof}
Case(i): $c+nw_1<\bar{T}<c+\bar{N}nw_1$.  Try candidate Lyapunov function $V(k)=(c+\nu(k)nw_1-\bar{T})^2/(nw_1)^2$.   Letting $\nu^*=(\bar{T}-c)/nw_1$ then this can be rewritten as $V(k)=(\nu(k)-\nu^*)^2$ and since $c+nw_1<\bar{T}<c+\bar{N}nw_1$ then $1<\nu^*<\bar{N}$.  In addition, to take care of gain $\gamma_0(k)$ we will show by induction that $\gamma_0(k)=1$.  By assumption $1<\bar{T}\nu(1)/(c+\nu(1)nw_1)\le \bar{N}$ and so $\gamma_0(1)=1$.   Suppose $\gamma_0(k)=1$. Substituting from (\ref{eq:insight}) it follows that
\begin{align*}
&V(k+1)=(\nu(k+1)-\nu^*)^2\\
&=(\max\{(1-K_2\frac{c+\nu(k)nw_1-\bar{T}}{c+\nu(k)nw_1})\nu(k),1\}-\nu^*)^2\\
&\stackrel{(a)}{\le}((1-K_2\frac{c+\nu(k)nw_1-\bar{T}}{c+\nu(k)nw_1})\nu(k)-\nu^*)^2\\
&=((1-K_2\frac{(c-\bar{T})/nw_1+\nu(k)}{c+\nu(k)nw_1}nw_1)\nu(k)-\nu^*)^2\\
&=(\nu(k)-K_2\frac{\nu(k)-\nu^*}{c+\nu(k)nw_1}nw_1\nu(k)-\nu^*)^2\\
&\stackrel{(b)}{=}(1-K_2\frac{\nu(k)nw_1}{c+\nu(k)nw_1})^2V(k)
\end{align*}
where $(a)$ follows because $\nu^*>1$.
Since $0<K_2<2$ and $0<\frac{\nu(k)nw_1}{c+\nu(k)nw_1}< 1$ it follows from $(b)$ that $0<(1-K_2\frac{\nu(k)nw_1}{c+\nu(k)nw_1})^2<1$ and so $V(k+1)$ is strictly decreasing unless $V(k)=0$.    Further, since $K_2<1$ then $\nu(k+1)$ has the same sign as $\nu(k)$ i.e. $\nu(k+1)>0$.  Putting these observations together, we have that $\nu(k+1)$ is closer than $\nu(k)$ to $\nu^*<\bar{N}$.   Since $\nu(1)\le \bar{N}$ then $\nu(2)<\bar{N}$, while when $\nu(k)\le \bar{N}$ for $k>1$ then $\nu(k+1)<\bar{N}$.  So by induction $\nu(k)\le\bar{N}$ for all $k\ge 1$ and thus $\gamma_0(k)=1$ for all $k\ge 1$.  Since $V(k+1)<V(k)$ when $V(k)>0$ then $V(k)$ decreases monotonically to $0$ i.e. the system converges to the point $c+\nu(k)nw_1-\bar{T}=0$ as claimed.

Cases (ii) and (iii). When $\bar{T}\ge c+\bar{N}nw_1$, respectively $\bar{T}<c+nw_1$, then $c+\nu(k)nw_1-\bar{T}<0$, respectively $c+\nu(k)nw_1-\bar{T}>0$ for all $1\le \nu(k)\le \bar{N}$.  The stated result now follows.
\end{proof}

Note that while the above analysis makes use of time-scale separation between $\vv{z}$ and $\nu$ so that $\vv{z}(k)=\nu(k)\vv{W}$, in practice we observe that the system is well behaved even when this assumption is violated and conjecture that Lemma \ref{lem:two} also applies in such cases.

\subsection{Selecting Controller Gains}
\subsubsection{Control Gain $K_1$}
\begin{figure}
\centering
\subfigure[$n=1$ station]{
\includegraphics[width=0.46\columnwidth]{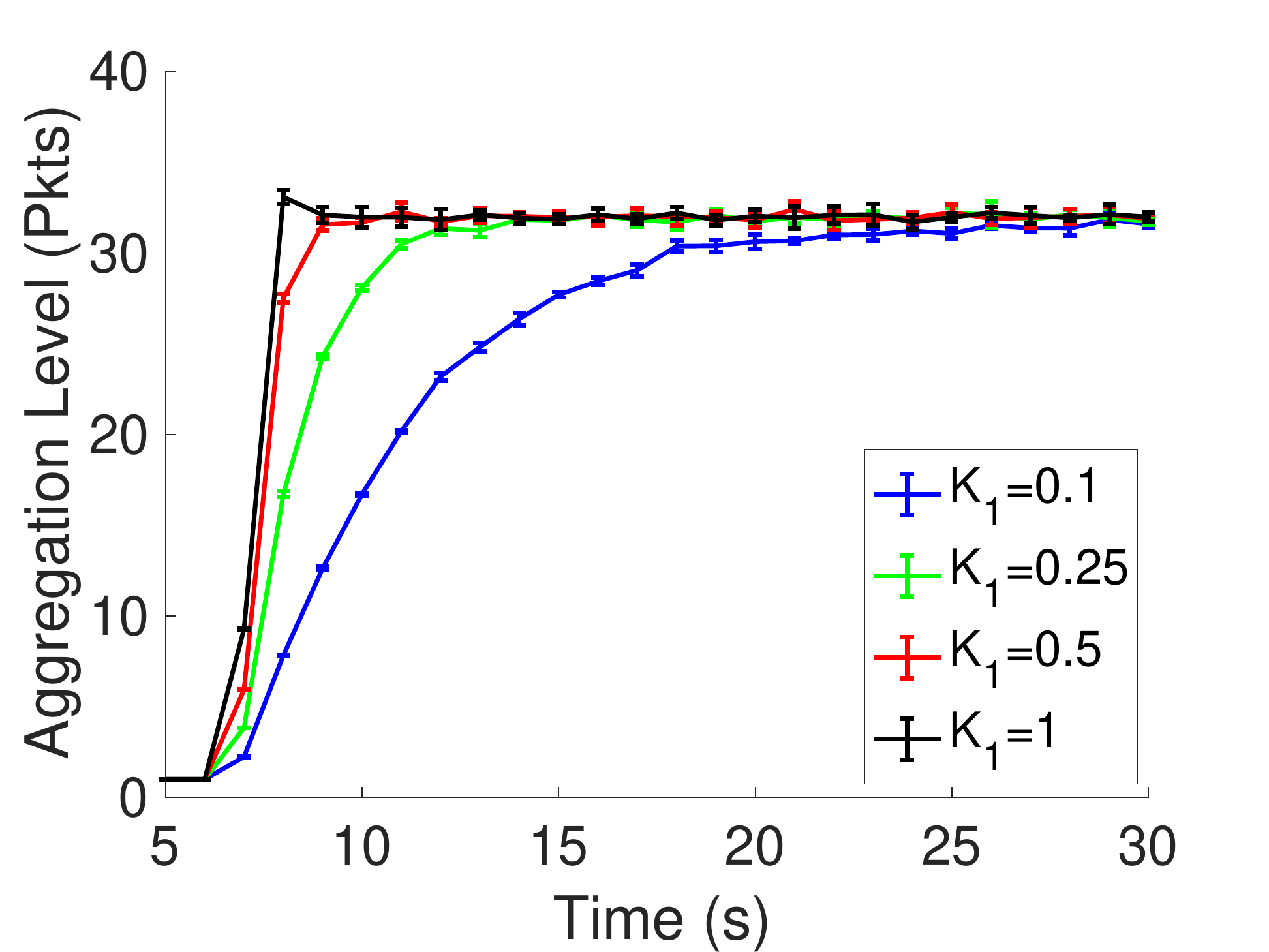}
}
\subfigure[$n=10$ stations]{
\includegraphics[width=0.46\columnwidth]{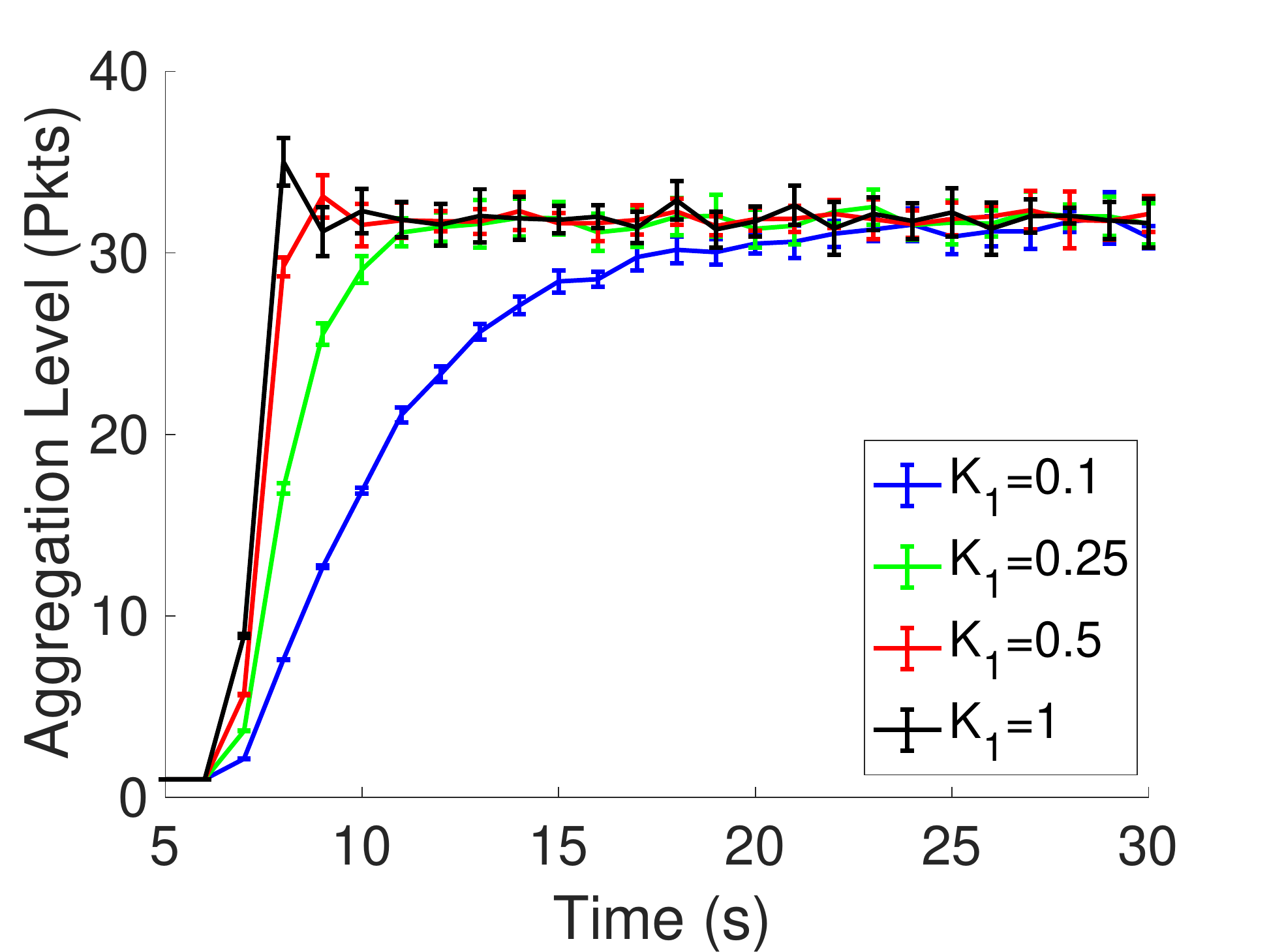}
}
\subfigure[$n=1$ station]{
\includegraphics[width=0.46\columnwidth]{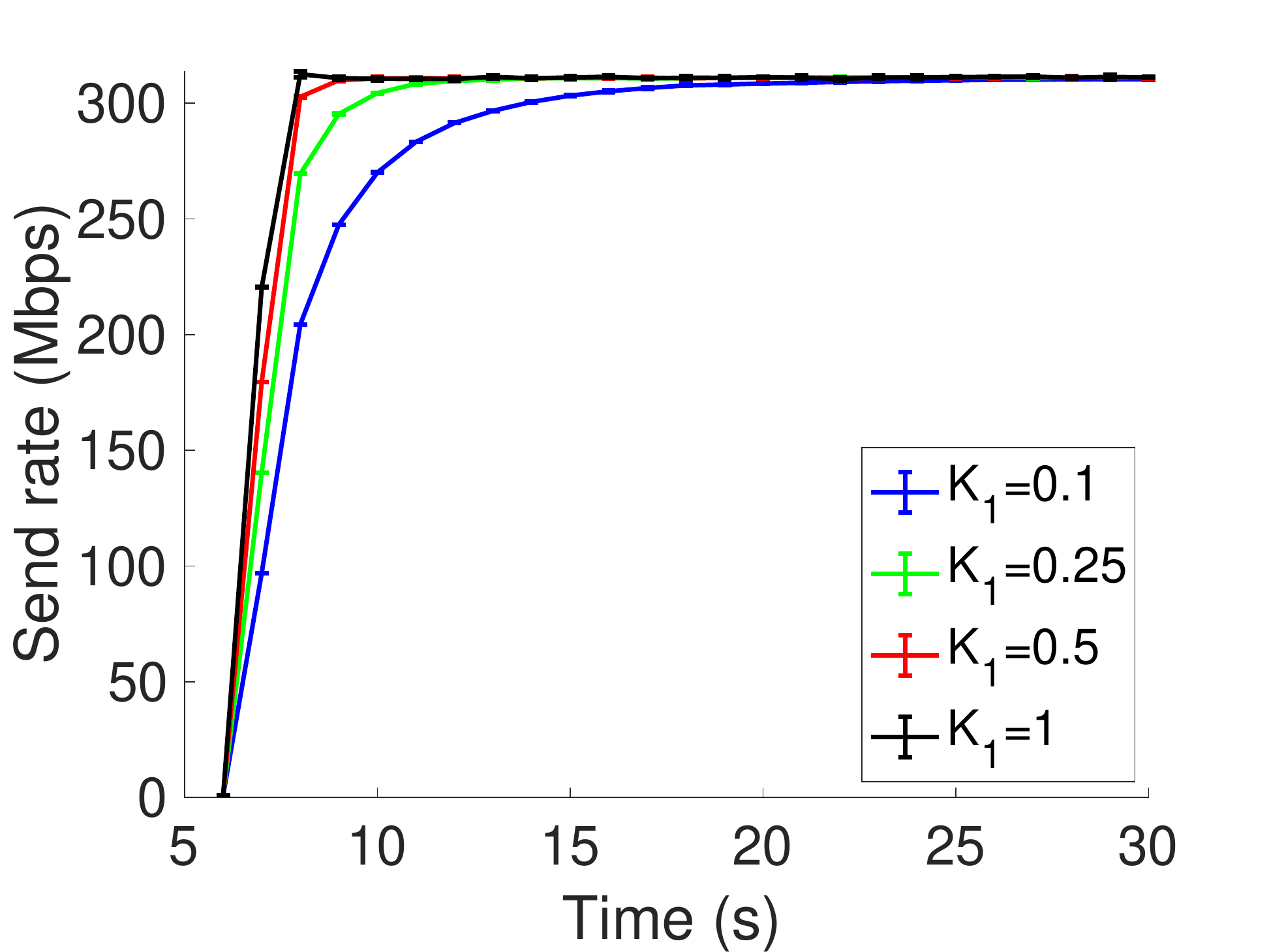}
}
\subfigure[$n=10$ stations]{
\includegraphics[width=0.46\columnwidth]{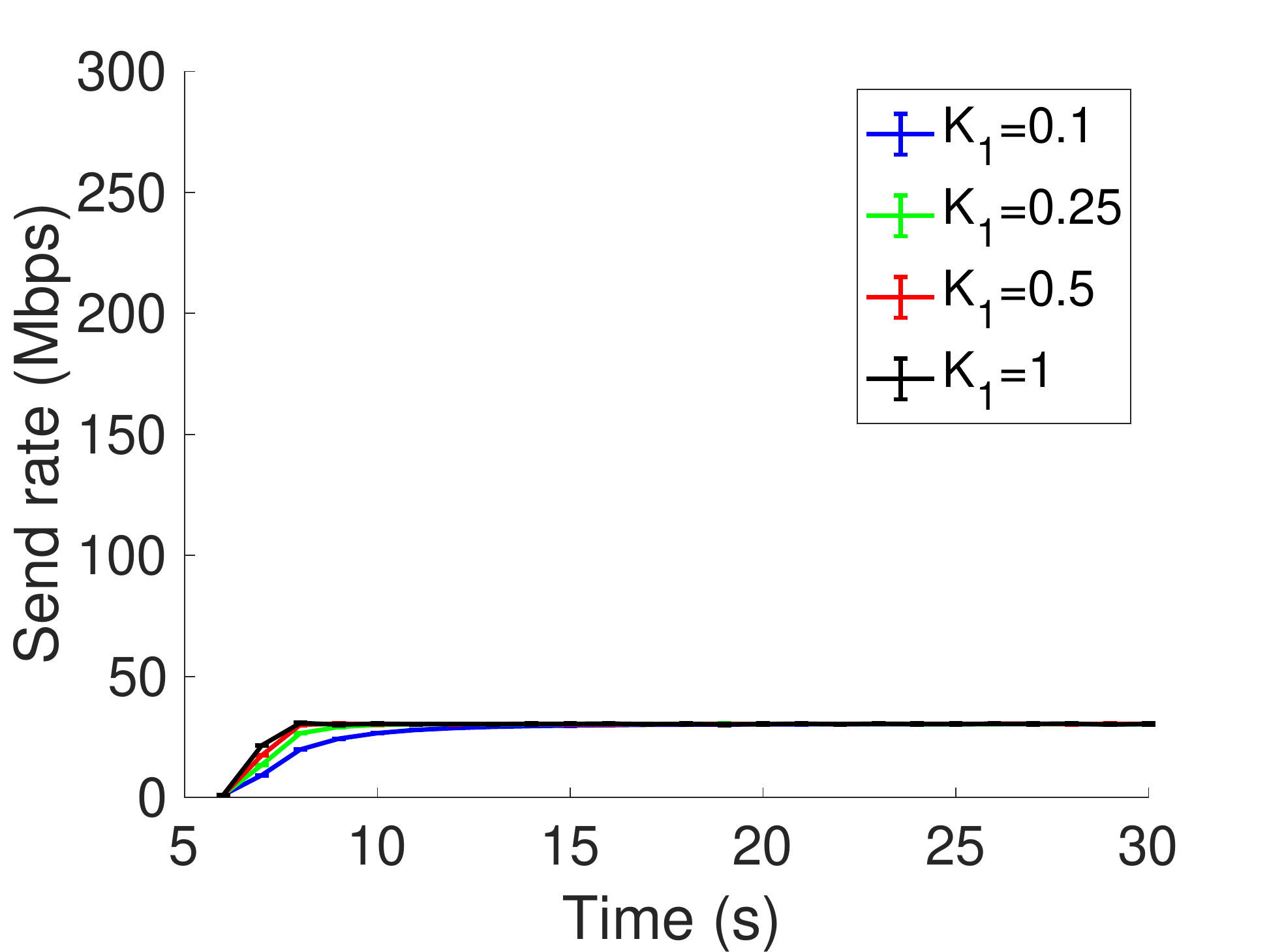}
}
\caption{Impact of control gain $K_1$ on transient dynamics of aggregation level and send rate.  Plots show average and standard deviation over 10 runs for each value of gain.  NS3 simulation, setup as in Section \ref{sec:expt}: $N_{target}=32$, NSS=1, MCS=9.  }\label{fig:two}
\end{figure}

Figure \ref{fig:two} plots the measured step response of the system aggregation level and send rate $x$ as the gain $K_1$ and number of stations $n$ are varied.  This data is for a detailed packet-level simulation, see Section \ref{sec:expt} for details.   It can be seen from Figures \ref{fig:two}(a)-(b) that, as expected, the aggregation level convergence time falls as $K_1$ is increased although the response starts to become oscillatory for larger values of $K_1$.   It can also be seen from these figures that that step response is effectively invariant with the number of stations due to the linearising action of the controller.    Figures \ref{fig:two}(c)-(d) show the send rate time histories corresponding to Figures \ref{fig:two}(a)-(b) and the impact of the nonlinearity $\tilde{F}_k$ relating aggregation level and send rate is evident with the send rate being an order of magnitude smaller for the same aggregation level with $n=10$ stations compared to with $n=1$ station.  Similar results are obtained when the MCS is varied.


 

In the remainder of this paper we select $K_1=0.5$ unless otherwise stated since this strikes a reasonable balance between response time and robustness to uncertainty in $c$ (with $K_1=0.5$ the value of $c$ can be out by a factor of 4, corresponding to a gain margin of 12 dB, and the system dynamics will remain stable).

%

\subsubsection{Control Gain $K_2$}
\begin{figure}
\centering
\subfigure[MCS 2, $\bar{T}=2.5ms$]{
\includegraphics[width=0.46\columnwidth]{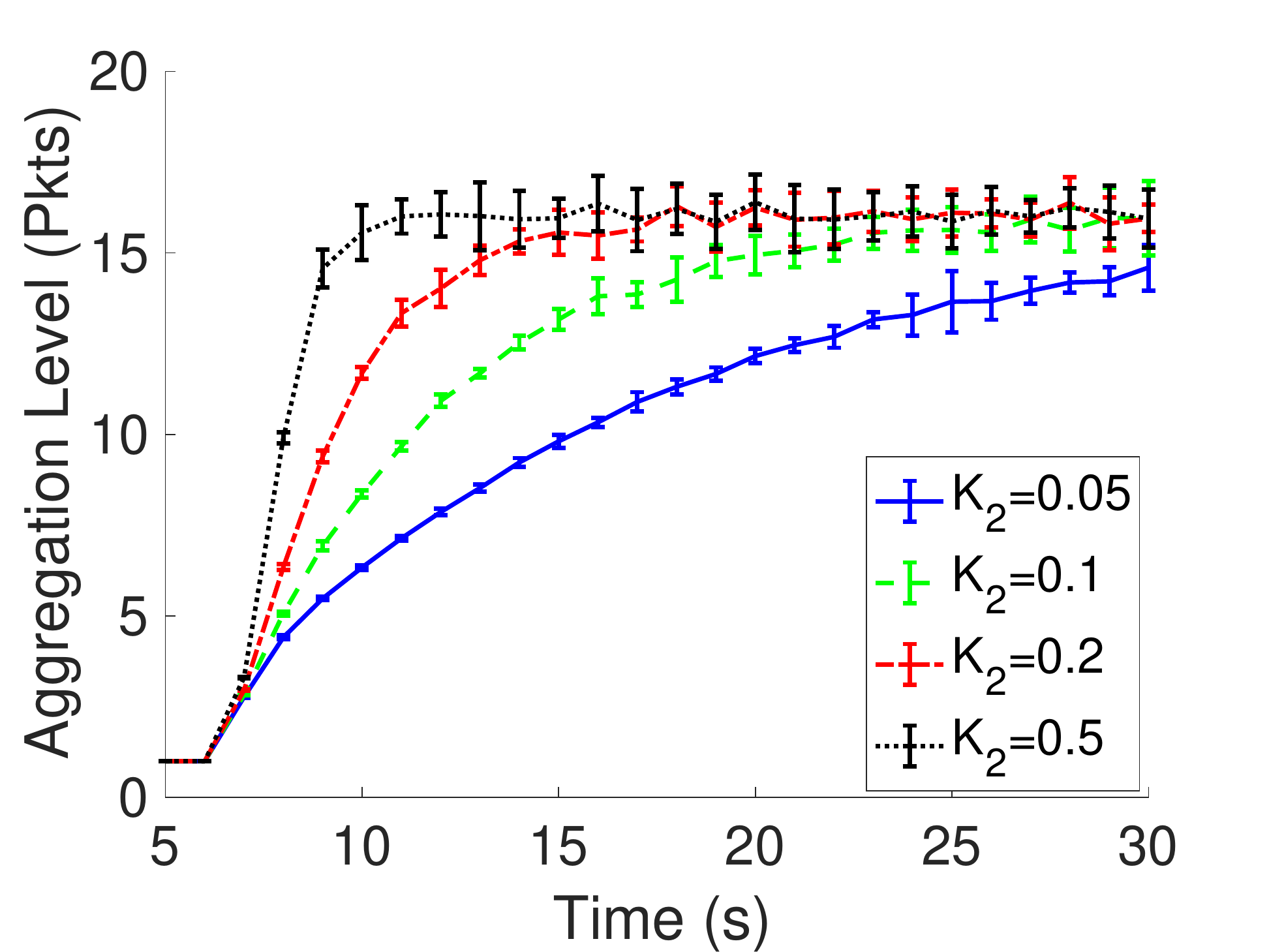}
}
\subfigure[$K_2=0.2$, $\bar{T}=2.5ms$]{
\includegraphics[width=0.46\columnwidth]{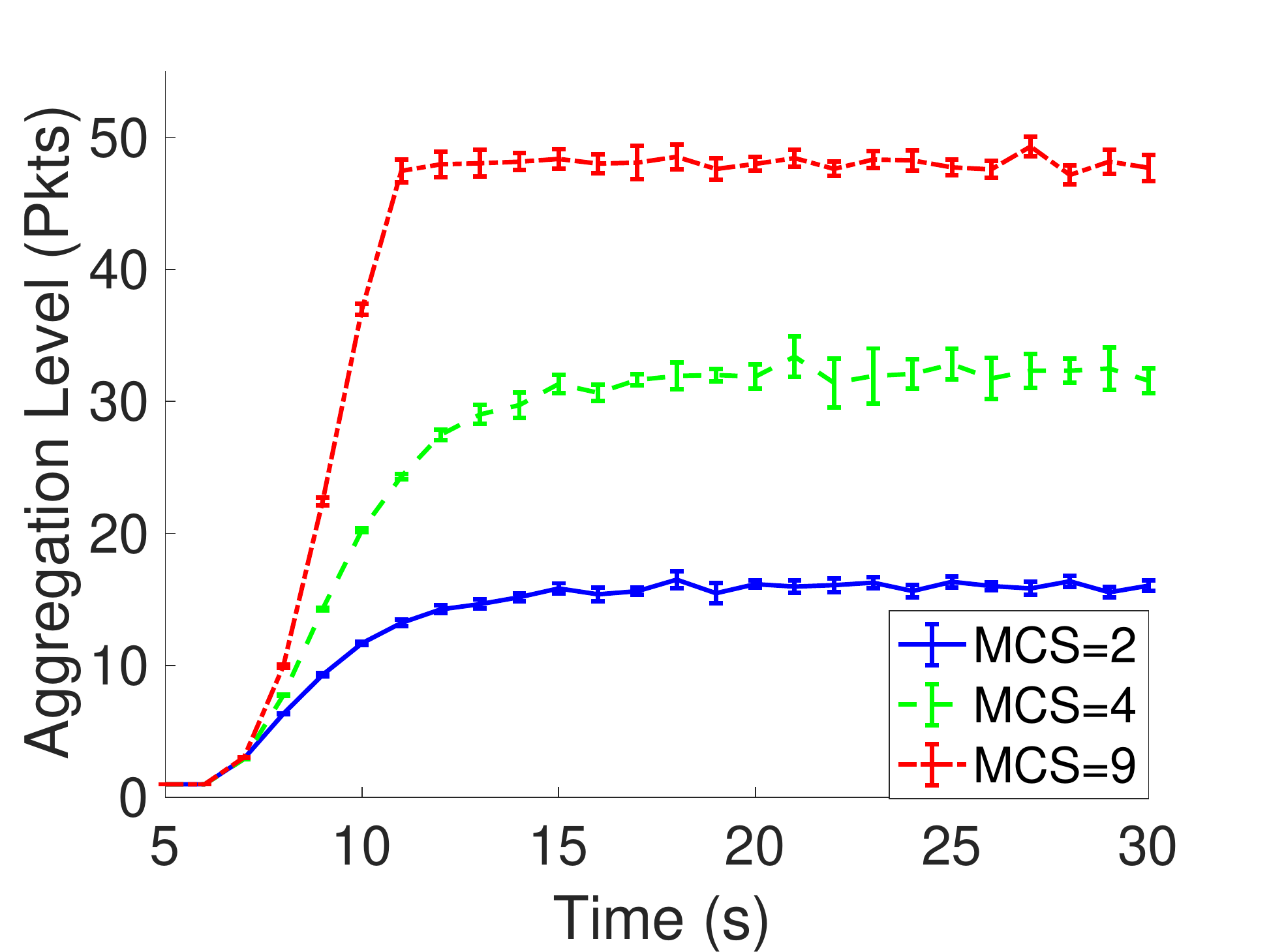}
}
\subfigure[$K_2=0.2$, $\bar{T}=2.5ms$]{
\includegraphics[width=0.46\columnwidth]{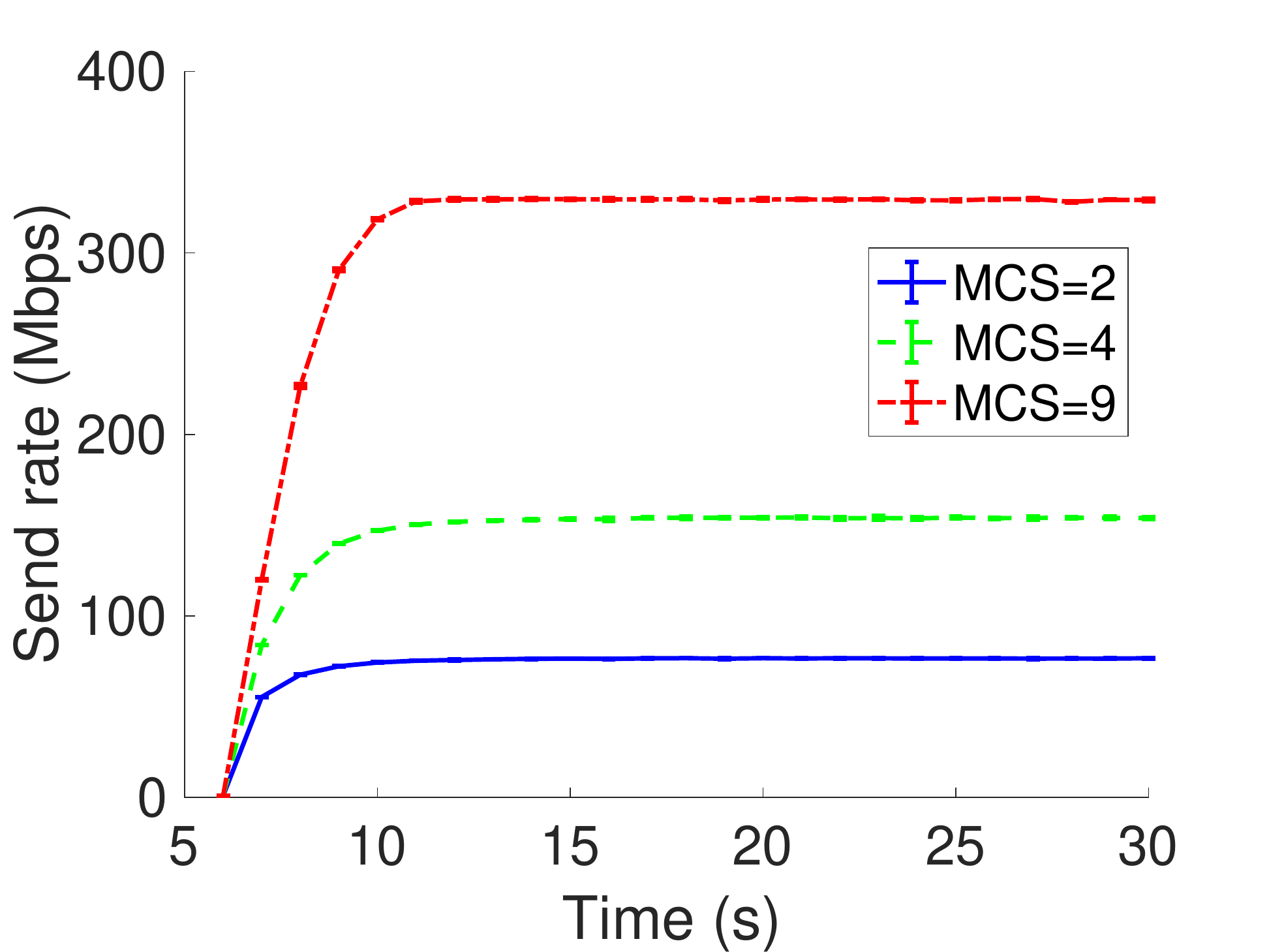}
}
\subfigure[$K_2=0.2$, $\bar{T}=2.5ms$]{
\includegraphics[width=0.46\columnwidth]{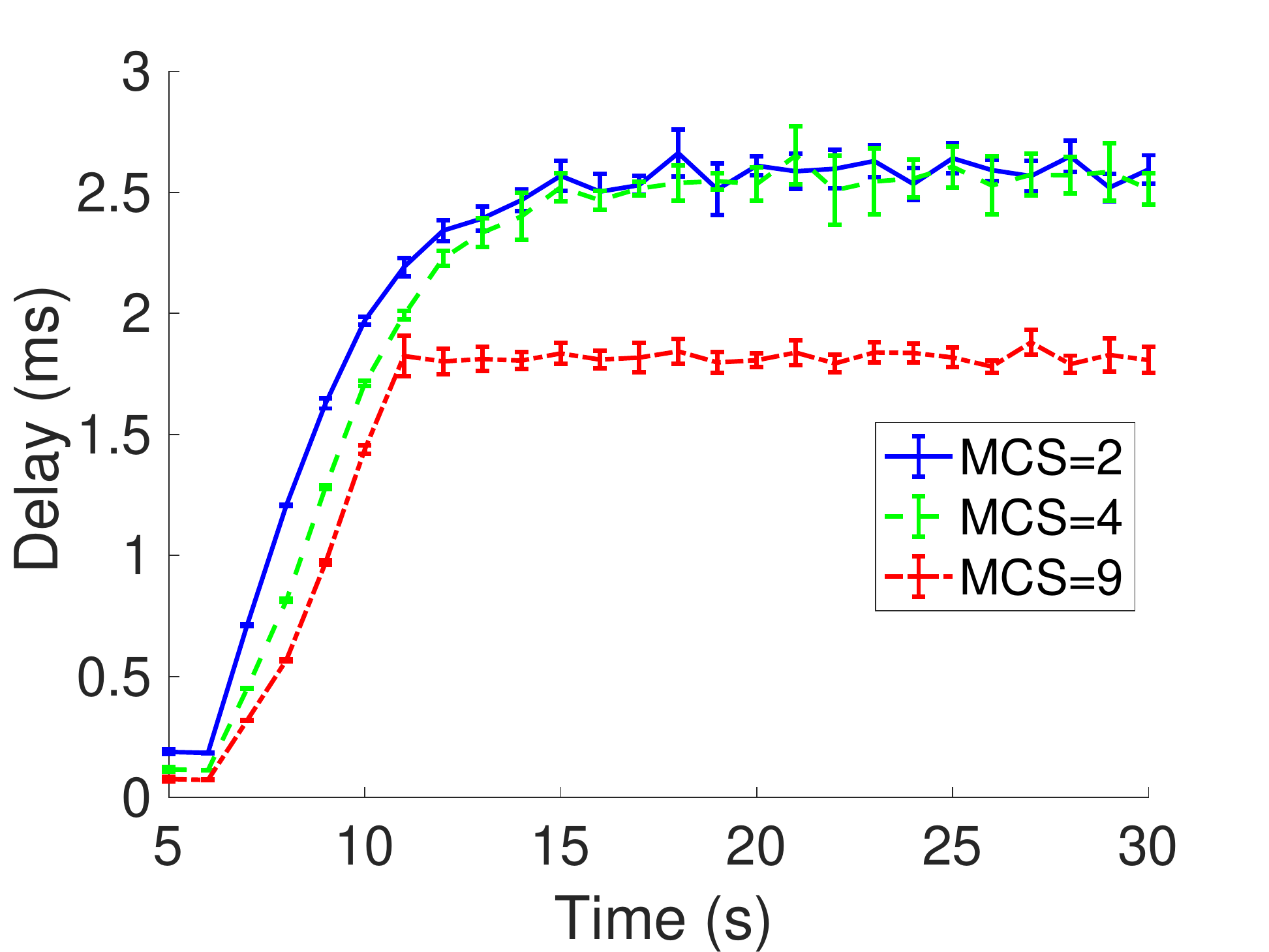}
}
\caption{(a) Impact of outer loop gain $K_2$ on convergence time, (b) adapting $N_{target}$ to regulate delay to below $\bar{T}$ as MCS is varied, (c), (d) send rate and delay measurements corresponding to (b).  Plots show average and standard deviation over 10 runs for each value of gain.   NS3, one client station, NSS=1, $\bar{T}=2.5ms$, $\bar{N}=48$.}\label{fig:three}
\end{figure}

Figure \ref{fig:three}(a) plots the measured step response of the system aggregation level as the outer control gain $K_2$ is varied.   It can be seen that the rise time falls with increasing gain, as expected.   Although not shown on the plot to reduce clutter, we observe that for $K_2\ge 1$ the response becomes increasing oscillatory suggesting that the sufficient condition for stability $K_2<1$ is in fact the stability boundary.    In the rest of the paper we select $K_2=0.2$ as striking a reasonable compromise between responsiveness and robustness to uncertainty.

Figures \ref{fig:three}(b)-(d) illustrate the adaptation by the outer feedback loop of $N_{target}$ so as to regulate the delay about the target value $\bar{T}$.   Figure \ref{fig:three}(b) plots the aggregation level vs time, Figure \ref{fig:three}(c) the send rate and Figure \ref{fig:three}(d) the delay.   Measurements are shown for three MCS values.  It can be seen that as the MCS rate increases both the aggregation level and send rate increase while the delay is maintained close to the target value $\bar{T}=2.5ms$\footnote{The 802.11ac standard imposes a maximum frame duration of 5.5ms.  In these tests with a single client station we use a target delay of 2.5ms so as to avoid hitting this upper limit on frame duration and thus allow the dynamics of the feedback loop to be seen more clearly.}, as expected.   

We can quickly verify the measurements as follows.  For the network configuration in Figures \ref{fig:three}(b)-(d) fixed overhead $c$ is around $200\mu$s.  MCS index 2 with NSS=1 corresponds data rate 87.7Mbps, the packet size $l=1500$B, overhead $l_{oh}=48$B and from Figure \ref{fig:three}(b) the aggregation level is approximately 16 packets, so $\vv{w}^T\mean{\vv{N}}{}=(1500+48)\times8\times16/87.7\times10^6=2.3$ms and adding $c$ to this gives $\bar{T}=2.5$ms.  Similarly, MCS index  4 with NSS=1 corresponds to a data rate of 175.5Mbps and plugging this value into the previous expression along with aggregation level 23 packets again gives $\vv{w}^T\mean{\vv{N}}{}=2.3$ms.   MCS index 9 with NSS=1 corresponds to data rate 390Mbps.   At this data rate we hit the limit $\bar{N}=48$ packets before delay target $\bar{T}$ is reached ($\vv{w}^T\mean{\vv{N}}{}=1.5$ms when the rate is 390Mbps and the aggregation level is 48 packets, adding $c=200\mu$s to this gives a delay of 1.7ms as can be seen in Figure \ref{fig:three}(d)). 

\subsubsection{Adapting $c$}
\begin{figure}
\centering
\subfigure[]{
\includegraphics[width=0.46\columnwidth]{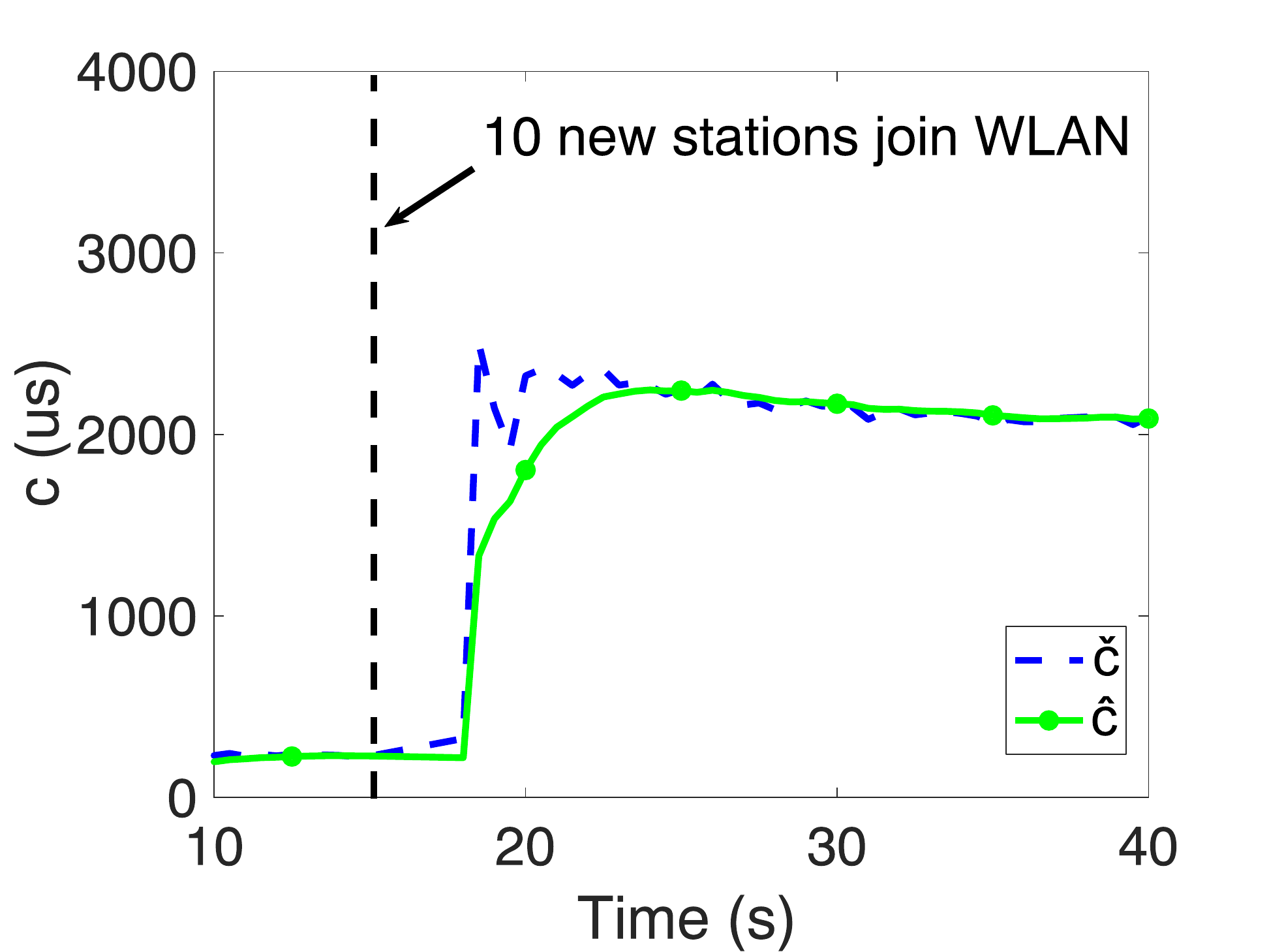}
}
\caption{Illustrating $c$ estimator (\ref{eq:cest}) tracking a sharp change in the number of stations from $n=1$ to $n=11$ at time 15s.  NS3 simulation, setup as in Section \ref{sec:expt}: one client station, $N_{target}=32$, NSS=1, MCS=9}\label{fig:disturb}
\end{figure}
The controller depends on parameter $c=n\mean{T}{_{oh}}$.  The average channel access time for each frame transmission is $CW/2\times S$ where $CW$ is the MAC contention window, typically 16 in 802.11ac, and $S$ is the MAC slot duration in seconds.  The PHY slot length is typically $9\mu s$, but the MAC slot duration can be significantly longer when other transmitters share the channel since the AP will defer access upon detecting the channel to be busy and it is this which makes it challenging to estimate $\mean{T}{_{oh}}$.  

An exact value for $c$ is not necessary since the feedback loop can compensate for uncertainty in $c$, i.e. an estimator that roughly tracks any large changes in $c$ is sufficient.  Recall that $\mean{N}{_i} = \frac{cx_i}{1-\vv{w}^T\vv{x}}$, i.e. $c=\frac{\mean{N}{_i}}{x_i}(1-\vv{w}^T\vv{x})$.  Motivated by this observation we use the following as an estimator of $c$,
\begin{align}
\hat{c}(k+1)=(1-\beta)\hat{c}(k) + \beta \check{c}(k)\label{eq:cest}
\end{align}
with $\check{c}(k):=\frac{\measmean{N}{_1}(k)}{x_1(k)}(1-\vv{w}^T\vv{x}(k))$, where $\beta$ is a design parameter which controls the window over which the moving average is calculated (a typical value is $\beta=0.05$).   

Figure \ref{fig:disturb}(b) illustrates the ability of this estimator to track a fairly significant change in the network conditions, namely 10 new stations joining the WLAN at time 15s and starting downlink transmissions.  These new stations cause a change in $c$ from a value of around $200\mu s$ to around $2200\mu s$ i.e. a change of more than an order of magnitude.  It can be seen that estimator (\ref{eq:cest}) tracks this large change without difficulty.   We observe similar tracking behaviour for changes in MCS and also when the channel is shared with other legacy WLANs.






\section{Experimental Measurements}\label{sec:expt}

\subsection{Hardware \& Software Setup}\label{sec:exptsetup}

\subsubsection{NS3 Simulator Implementation}\label{sec:ns3}
We  implemented the inner-outer controller in the NS3 packet-level simulator.  Based on the received feedbacks it periodically configures the sending rate of {\tt udp-client} applications colocated at a single node connected to an Access Point. Each wireless station receives a single UDP traffic flow at a {\tt udp-server} application that we modified to collect frame aggregation statistics and periodically transmit these to the controller at intervals of $\Delta$ ms.  We also developed a round-robin scheduler at the AP with separate queue for each destination, and we added new functions to let stations determine the MCS of each received frame together with the number of MPDU packets it contains.  The maximum aggregation level permitted is $N_{max}$=64.  We configured 802.11ac to use a physical layer operating over an $80MHz$ channel, VHT rates for data frames and legacy rates for control frames.  The PHY MCS and the number of spatial streams NSS used can be adjusted.  As validation we reproduced a number of the simulation measurements in our experimental testbed and found them to be in good agreement. {The new NS3 code and the software that we used to perform experimental evaluations are available open-source}\footnote{Code can be obtained by contacting the corresponding author.}. 

\subsubsection{Experimental Testbed}
Our experimental testbed uses an Asus RT-AC86U Access Point (which uses a Broadcom 4366E chipset and supports 802.11ac MIMO with up to three spatial streams.   It is configured to use the 5GHz frequency band with 80MHz channel bandwidth.   This setup allows high spatial usage (we observe that almost always three spatial streams are used) and high data rates (up to MCS 9).    By default aggregation supports AMSDU's and allows up to 128 packets to be aggregated in a frame (namely 64 AMSDUs each containing two packets).  

A Linux server connected to this AP via a gigabit switch uses iperf 2.0.5 to generate UDP downlink traffic to the WLAN clients.     Iperf inserts a sender-side timestamp into the packet payload and since the various machines are tightly synchronised over a LAN this can be used to estimate the one-way packet delay (the time between when a packet is passed into the socket in the sender and when it is received).  Note, however, that in production networks accurate measurement of one-way delay is typically not straightforward as it is difficult to maintain accurate synchronisation between server and client clocks (NTP typically only synchronises clocks to within a few tens of milliseconds).  

\subsection{Simulation Measurements}
Figure \ref{fig:ns3results} plots measured simulation performance of the controller as the target delay $\bar{T}$ is varied from 5-20ms, the number of client stations is varied from 1 to 25 and for two values of MCS rate.  It can be seen that the controller consistently regulates the delay quite tightly around the target value except when the aggregation level hits the specified upper limit of $48$ packets, as expected.   Also shown on these plots are the 75th percentile values.  These mostly overlay the mean values, indicating tight regulation of delay and rate.  


\begin{figure}
\centering
\subfigure[MCS 9, NSS 1 (390Mbps)]{
\includegraphics[width=0.46\columnwidth]{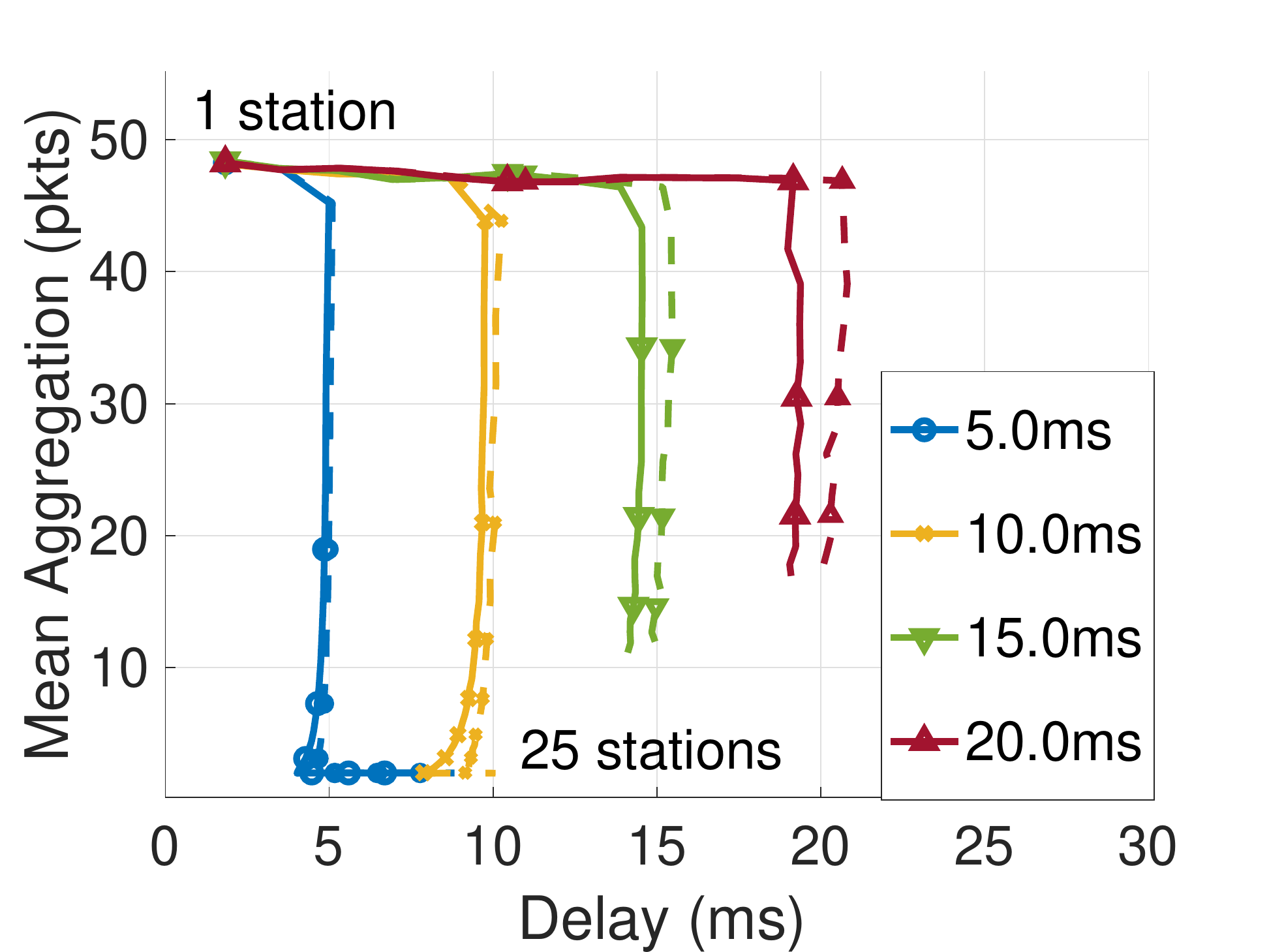}
}
\subfigure[]{
\includegraphics[width=0.46\columnwidth]{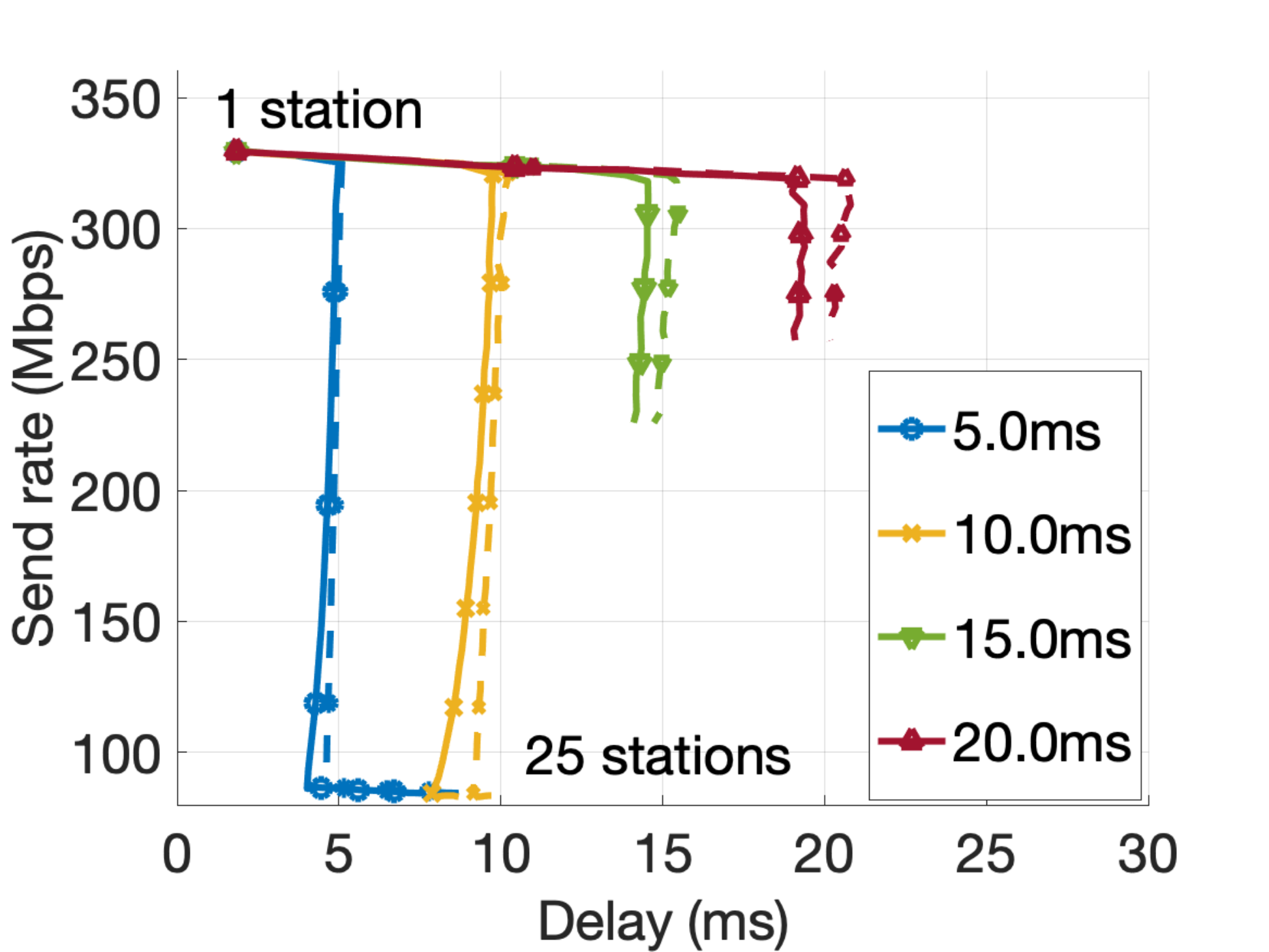}
}
\subfigure[MCS 4, NSS 1 (175Mbps)]{
\includegraphics[width=0.46\columnwidth]{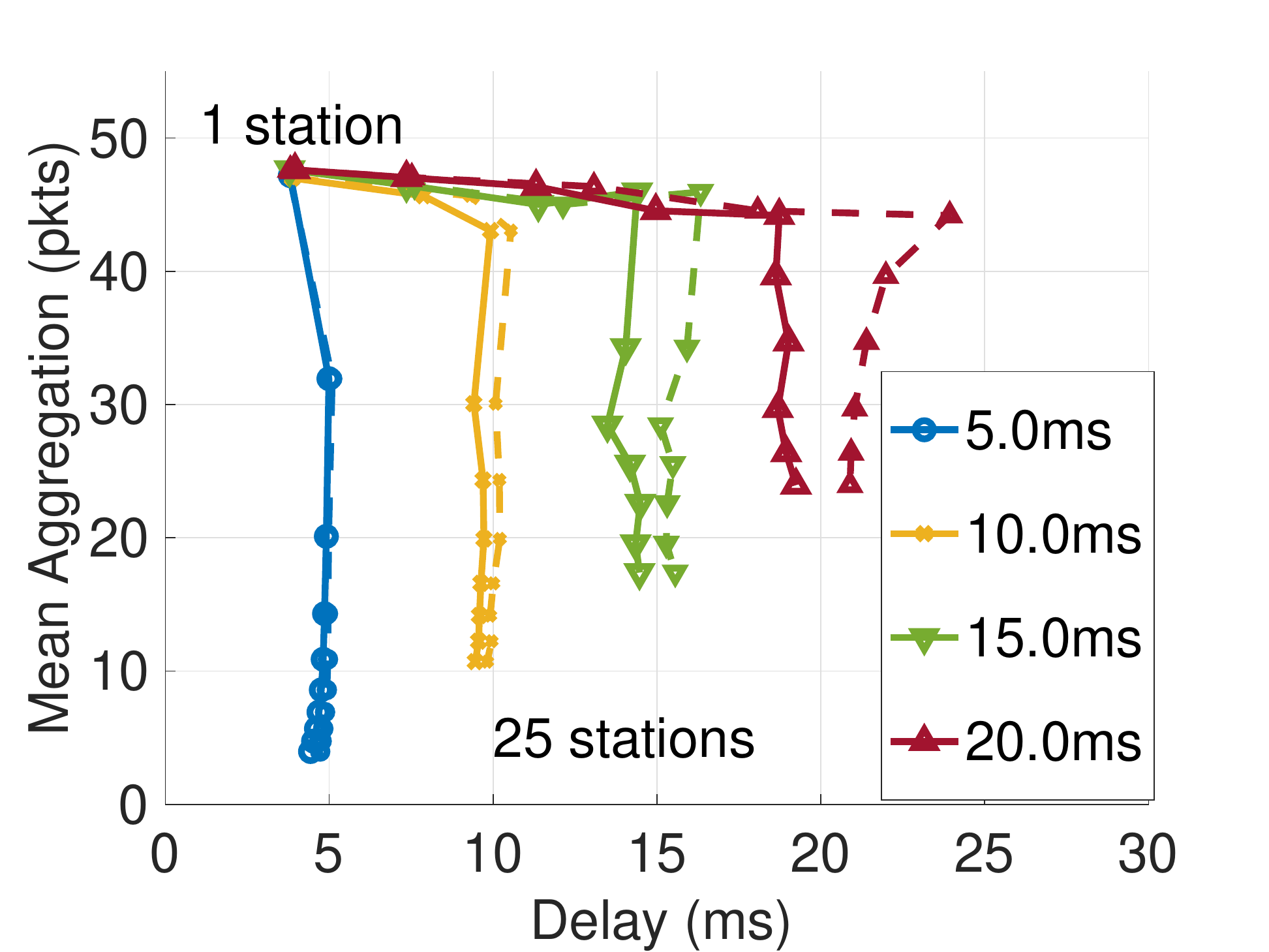}
}
\subfigure[]{
\includegraphics[width=0.46\columnwidth]{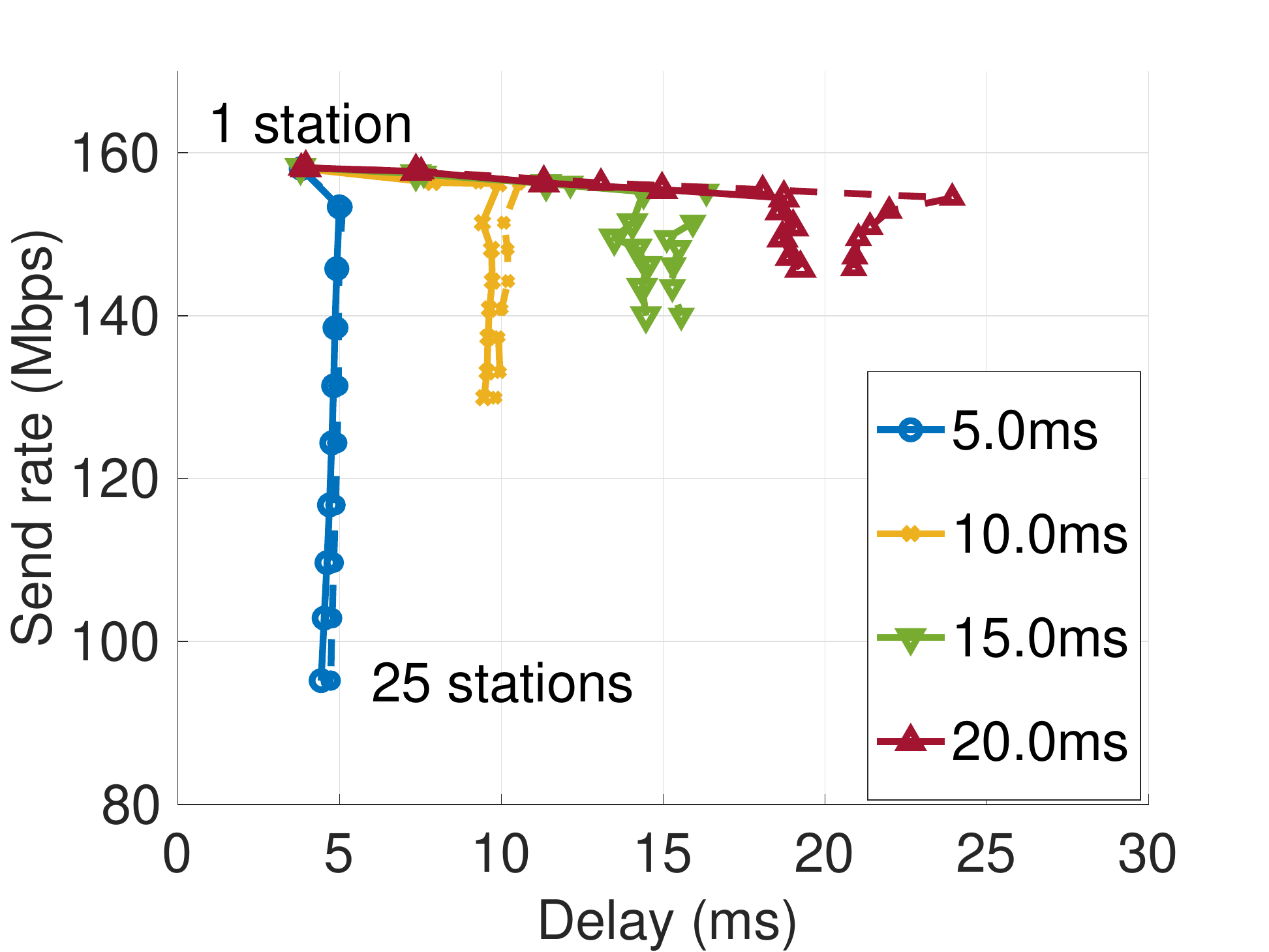}
}
\caption{Measured delay and rate of the controlled system as the target delay, number of client stations and MCS rate are varied.  Solid lines indicate mean delay and rate, dashed lines the 75th percentile values. $\bar{N}=48$ packets, NS3 data.}\label{fig:ns3results}
\end{figure}

\subsection{Experimental Measurements}\label{sec:expts}

\begin{figure}
\centering
\subfigure[Receive Rate]{
\includegraphics[width=0.46\columnwidth]{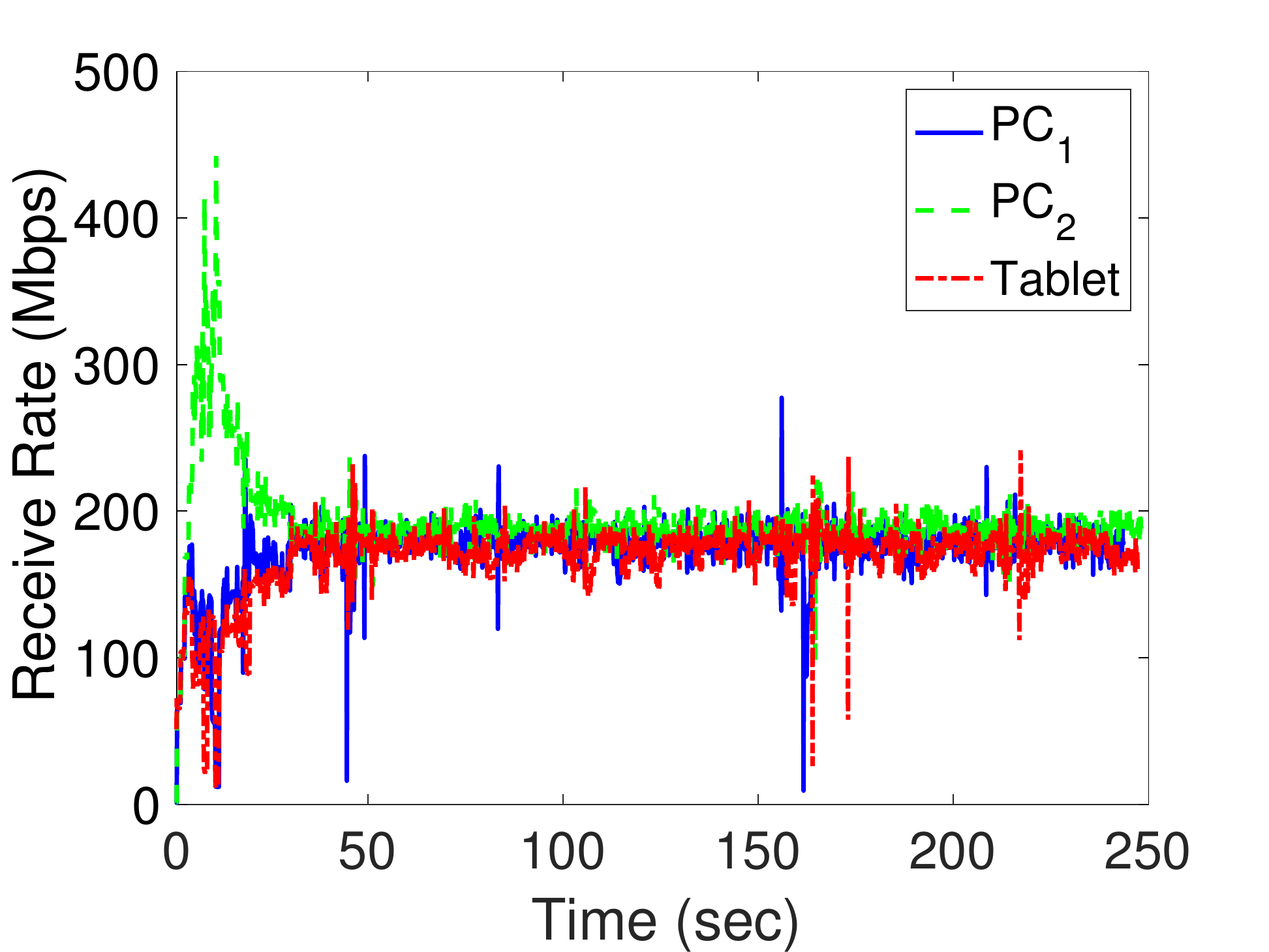}
}
\subfigure[One-way Delay]{
\includegraphics[width=0.46\columnwidth]{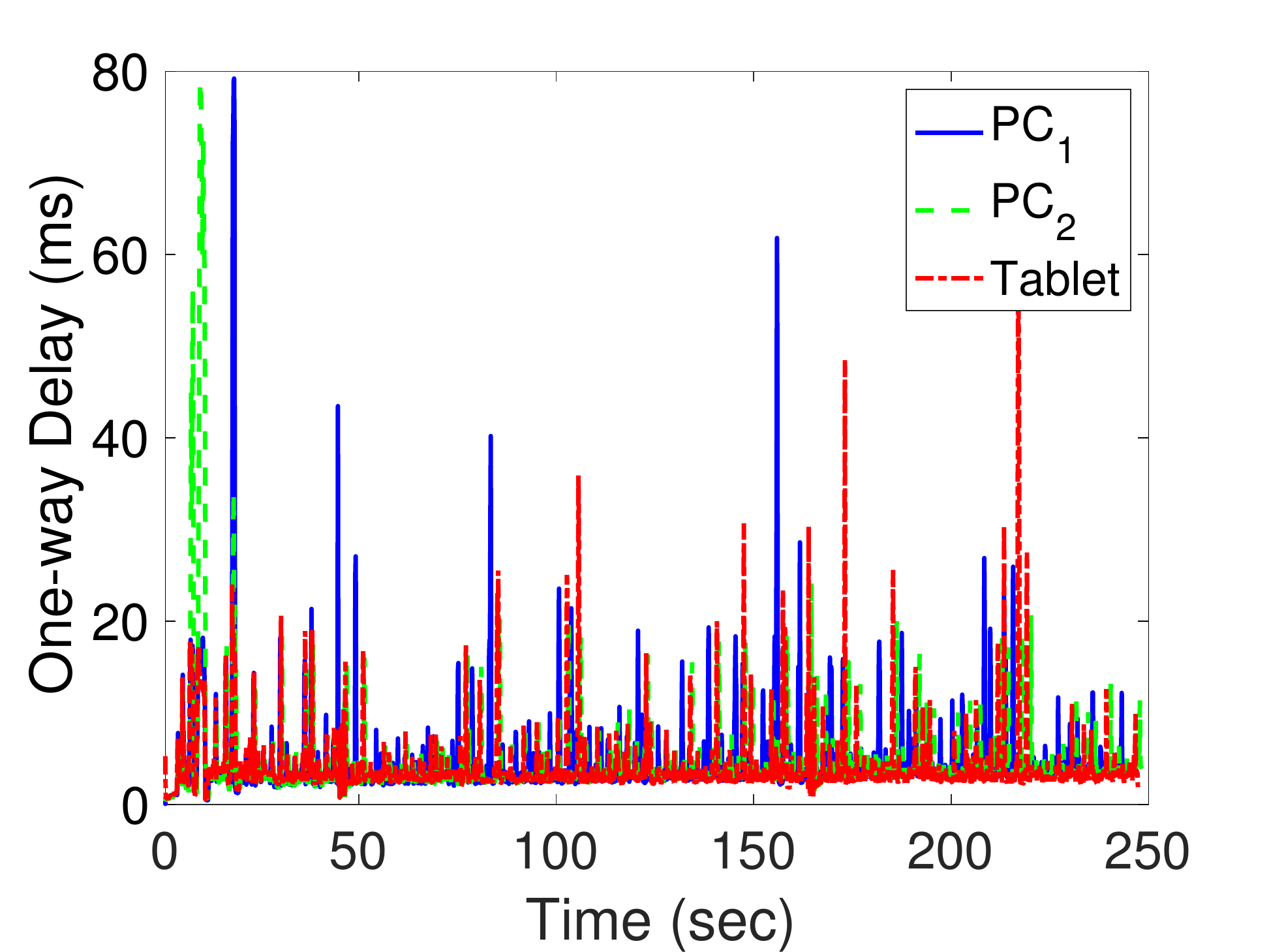}
}
\subfigure[MCS]{
\includegraphics[width=0.46\columnwidth]{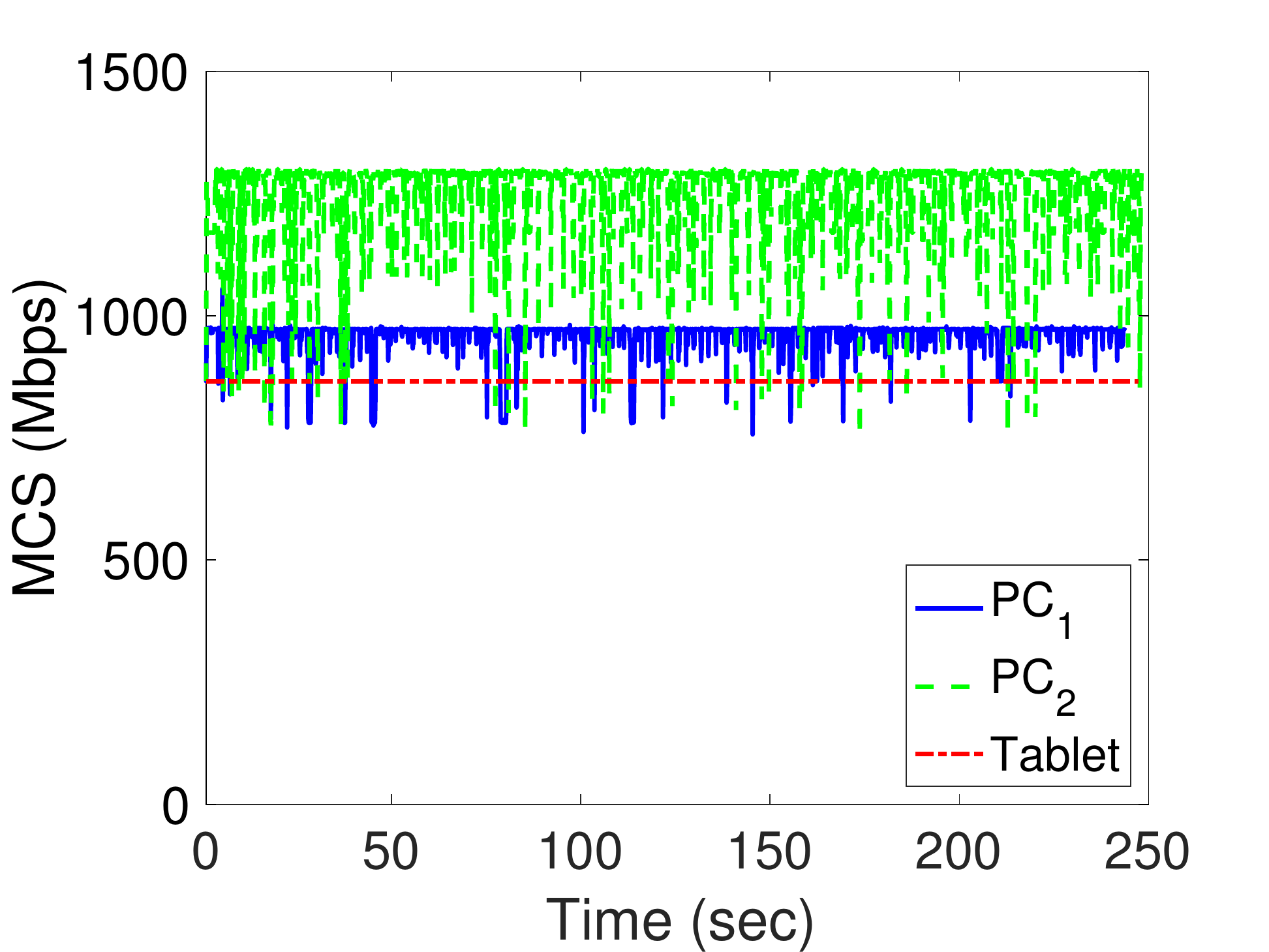}
}
\subfigure[Aggregation Level]{
\includegraphics[width=0.46\columnwidth]{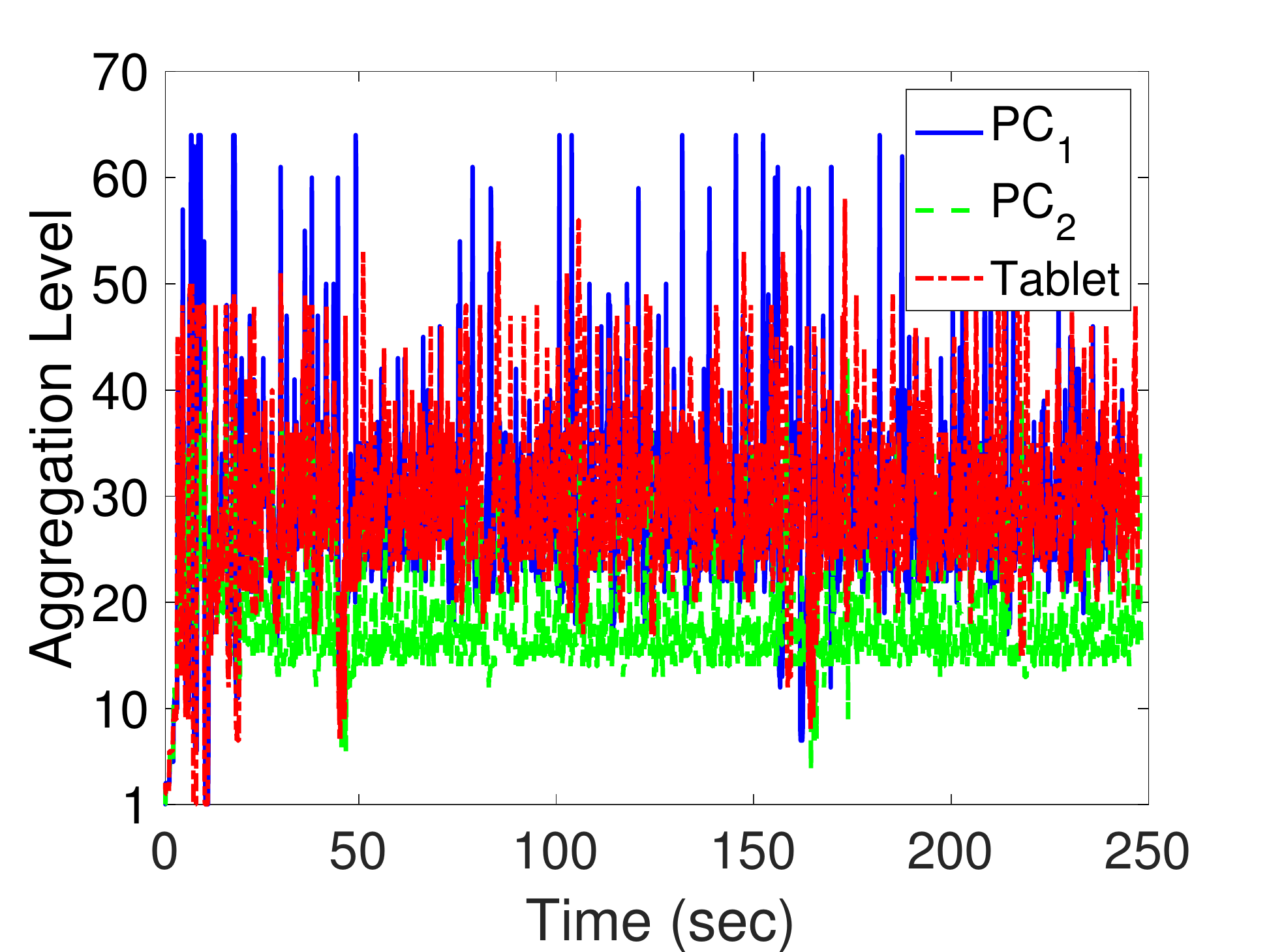}
}
\caption[Managing an edge network using the nonlinear feedback controller.]{Managing an edge network using the nonlinear feedback controller. The one-way delay and MCS values are averaged over 100ms intervals. 
Experimental data.}\label{fig:testbed2}
\end{figure}

\subsubsection{Single Station}
Figure \ref{fig:edge}(b) plots typical a rate and delay time-history measured in our experimental testbed with a single client station.  The rate is close to the maximum capacity while the delay is maintained at a low value of around 2ms.   

\subsubsection{Multiple Stations}
Figure \ref{fig:testbed2} plots measured time histories with three client stations (two PCs and an android tablet).   It can be seen from Figure \ref{fig:testbed2}(a) that the rates to the stations quickly converge.  Figure \ref{fig:testbed2}(b) shows the corresponding delays, which are regulated around the target value of $\bar{T} = 10$ms, although fluctuations due to MAC and channel randomness can also be seen.  Figure \ref{fig:testbed2}(c) shows the measured MCS rates of the three stations, which reflect the radio channel quality (a higher rate indicating a better channel) and it can seen that PC2 has a significantly better channel than the two other clients (it is located closer to the AP).   Figure \ref{fig:testbed2}(d)  shows the measured aggregation levels, and since an equal airtime policy is enforced by the controller it can be seen that the aggregation level of PC2 is lower (since its MCS rate is higher).

\begin{figure}
 \centering
 \subfigure[Receive Rate]{
 \includegraphics[width=0.46\columnwidth]{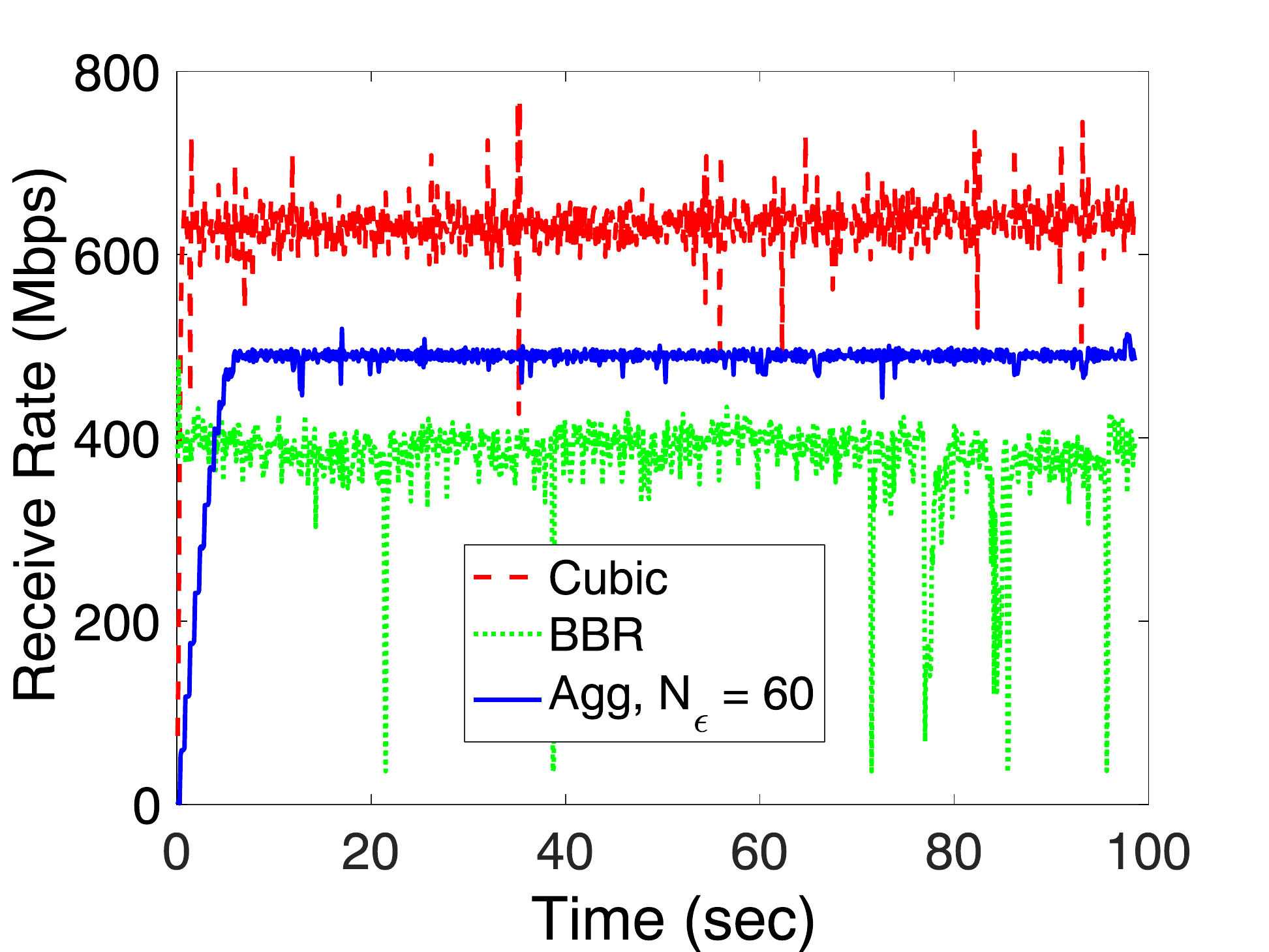}
 }
\subfigure[One-way Delay]{
 \includegraphics[width=0.46\columnwidth]{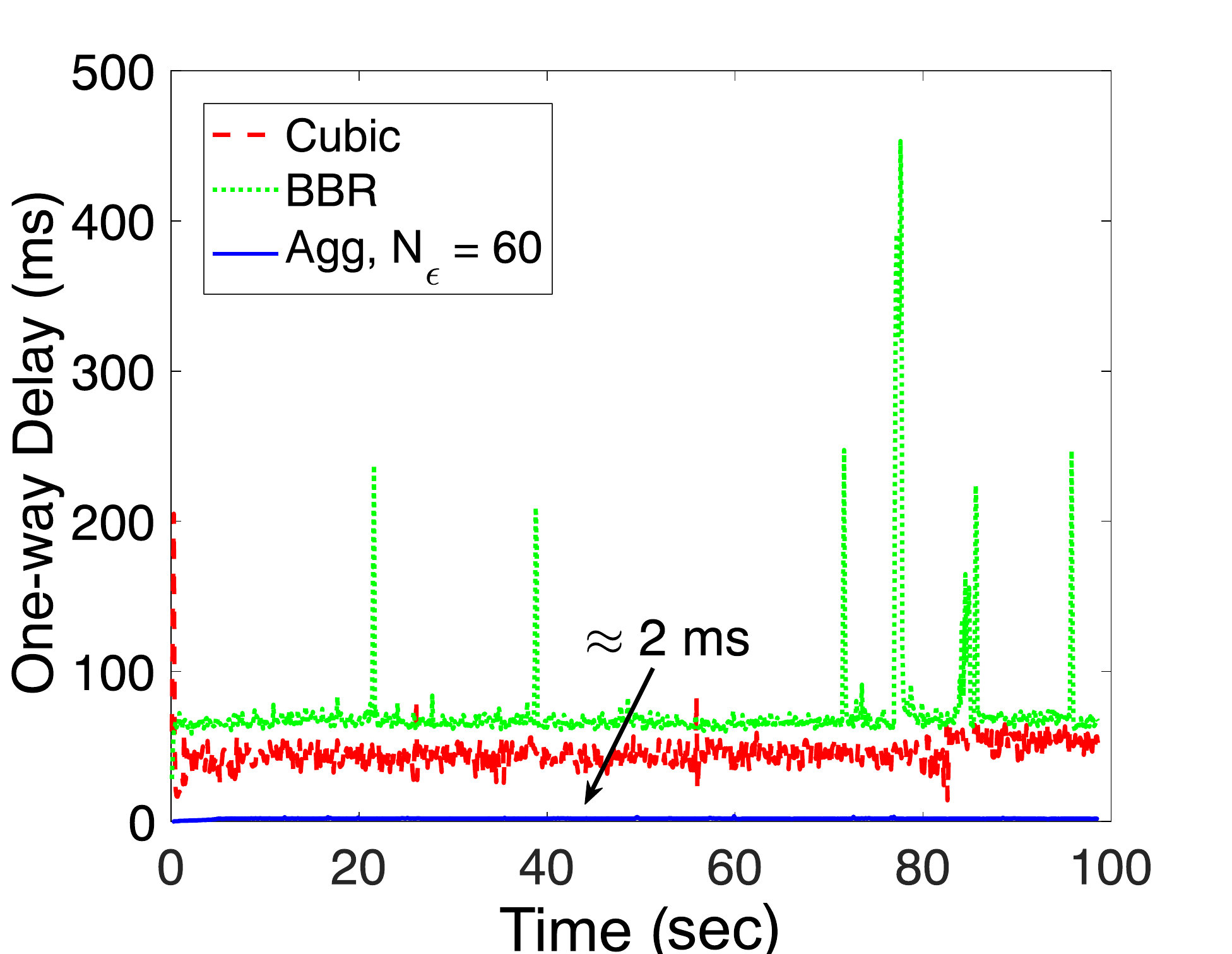}
 }
 \subfigure[\#Packet loss]{
 \includegraphics[width=0.46\columnwidth]{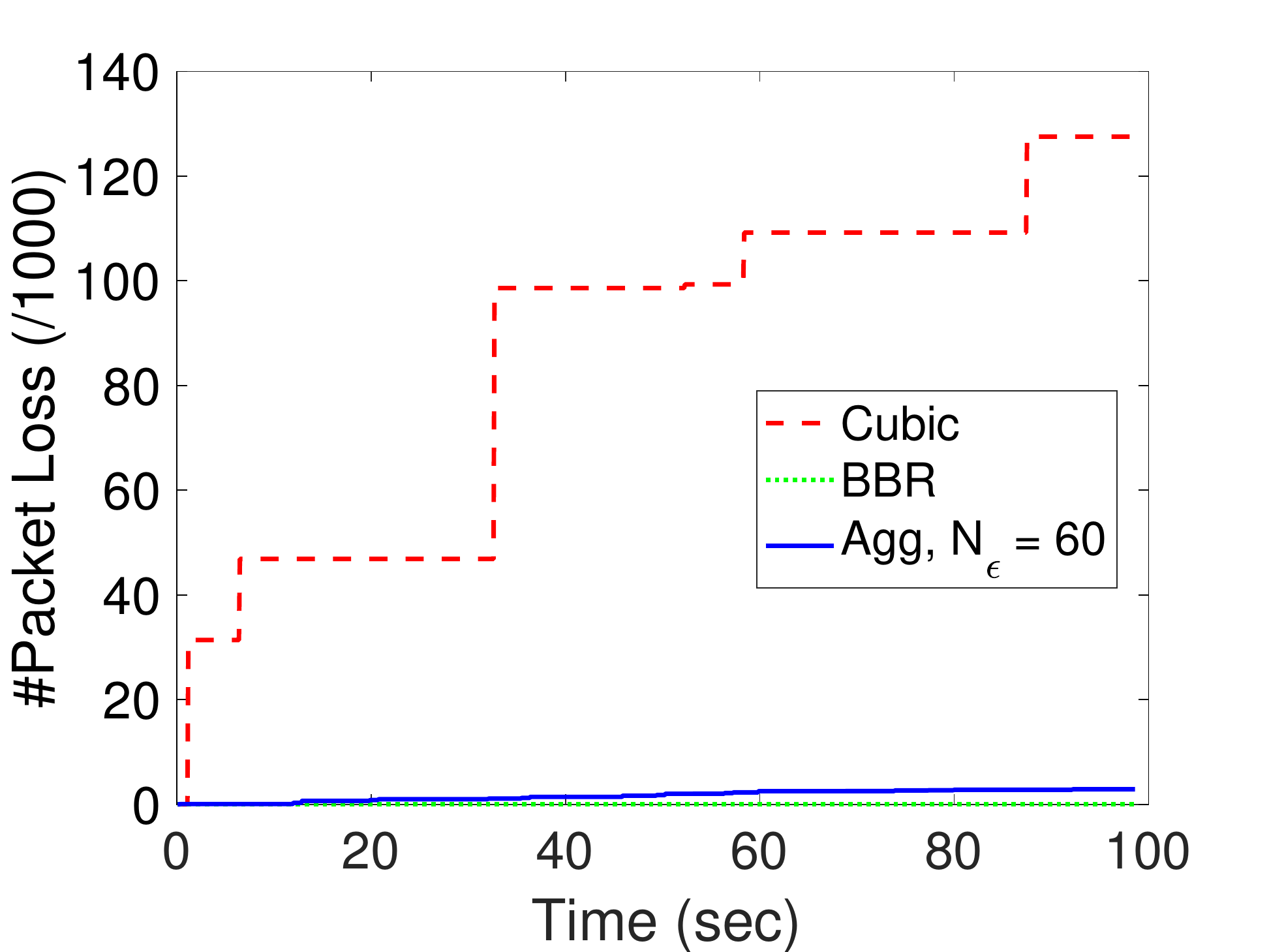}
 }
 \caption{Compare the performance of aggregation-based rate control algorithm with TCP Cubic and BBR.   The one-way delay in (b) is averaged over 100ms intervals.  
 Experimental data.}\label{fig:fig23}
\end{figure}

\subsubsection{Performance Comparison With TCP Cubic \& BBR}
We briefly compare the performance of the aggregation-based rate control algorithm with TCP Cubic~\cite{Cubic}, the default congestion control algorithm used by Linux and Android.  In addition, we compare performance against TCP BBR \cite{BBR} since this is a state-of-the-art congestion control algorithm currently being developed by Google and which also targets high rate and low latency.

Since TCP Cubic is implemented on Android we use a Samsung Galaxy tablet as client. However, TCP BBR is not currently available for Android and so we use a Linux box (Debian Stretch, 4.9.0-7-amd64 kernel) as the BBR client. 

Figure \ref{fig:fig23} shows typical receive rate and one-way delay time histories measured for the three algorithms.  It can be seen from Figure \ref{fig:fig23}(a) that Cubic selects the highest rate (around 600Mbps) but from Figure \ref{fig:fig23}(b) that this comes at the cost of high one-way delay (around 50ms).  This is as expected since Cubic uses loss-based congestion control and so increases the send rate until queue overflow (and so a large queue backlog and high queueing delay at the AP) occurs.   As confirmation, Figure \ref{fig:fig23}(c) plots the number of packet losses vs time and it can be seen that these increase over time when using Cubic, each step increase corresponding to a queue overflow event followed by backoff of the TCP congestion window.

BBR selects the lowest rate (around 400Mbps) of the three algorithms, but surprisingly also has the highest end-to-end one-way delay (around 75ms).  High delay when using BBR has also previously been noted by e.g. \cite{Parisa} where the authors propose that high delay is due to end-host latency within the BBR kernel implementation at both sender and receiver.   However, since our focus is not on BBR we do not pursue this further here but note that the BBR Development team at Google is currently developing a new version of BBR v2.

 Our low delay aggregation-based approach selects a rate (around 480 Mbps), between that of Cubic and BBR, consistent with the analysis in earlier sections.   Importantly, the end-to-end one-way delay is around 2ms i.e. more than 20 times lower than that with Cubic and BBR.   It can also be seen from Figure \ref{fig:fig23}(c) that it induces very few losses (a handful out of the around 4M packets sent over the 100s interval shown).
 
%
%

\section{Conclusions}
In this paper we consider the analysis and design of a feedback controller to regulate queueing delay in a next generation edge transport architecture for 802.11ac WLANs.   We develop a novel nonlinear control design inspired by the solution to an associated proportional fair optimisation problem.   The controller compensates for plant nonlinearities and so can be used for the full envelope of operation.   The robust stability of the closed-loop system is analysed and the selection of control design parameters discussed.     We develop an implementation of the nonlinear control design and use this to present a performance evaluation using both simulations and experimental measurements.

\section*{Acknowledgements}
DL would like to thank Hamid Hassani for discussions and carrying out the experimental tests in Section \ref{sec:expts}.

\bibliographystyle{IEEEtran}
\bibliography{references.bib}

\begin{thebibliography}{10}
\providecommand{\url}[1]{#1}
\csname url@samestyle\endcsname
\providecommand{\newblock}{\relax}
\providecommand{\bibinfo}[2]{#2}
\providecommand{\BIBentrySTDinterwordspacing}{\spaceskip=0pt\relax}
\providecommand{\BIBentryALTinterwordstretchfactor}{4}
\providecommand{\BIBentryALTinterwordspacing}{\spaceskip=\fontdimen2\font plus
\BIBentryALTinterwordstretchfactor\fontdimen3\font minus
  \fontdimen4\font\relax}
\providecommand{\BIBforeignlanguage}[2]{{%
\expandafter\ifx\csname l@#1\endcsname\relax
\typeout{** WARNING: IEEEtran.bst: No hyphenation pattern has been}%
\typeout{** loaded for the language `#1'. Using the pattern for}%
\typeout{** the default language instead.}%
\else
\language=\csname l@#1\endcsname
\fi
#2}}
\providecommand{\BIBdecl}{\relax}
\BIBdecl

\bibitem{quickandplenty}
F.~G. Hamid~Hassani and D.~J. Leith, ``{{Quick and Plenty: Achieving Low Delay
  and High Rate in 802.11ac Edge Networks}},'' \emph{Computer Networks}, vol.
  187, no.~14, 2021.

\bibitem{BanchsSA06}
\BIBentryALTinterwordspacing
A.~Banchs, P.~Serrano, and A.~Azcorra, ``End-to-end delay analysis and
  admission control in 802.11 {DCF} wlans,'' \emph{Computer Communications},
  vol.~29, no.~7, pp. 842--854, 2006. [Online]. Available:
  \url{https://doi.org/10.1016/j.comcom.2005.08.006}
\BIBentrySTDinterwordspacing

\bibitem{BoggiaCGM07}
\BIBentryALTinterwordspacing
G.~Boggia, P.~Camarda, L.~A. Grieco, and S.~Mascolo, ``Feedback-based control
  for providing real-time services with the 802.11e {MAC},'' \emph{{IEEE/ACM}
  Trans. Netw.}, vol.~15, no.~2, pp. 323--333, 2007. [Online]. Available:
  \url{http://doi.acm.org/10.1145/1279660.1279666}
\BIBentrySTDinterwordspacing

\bibitem{Garcia-SaavedraBSW12}
\BIBentryALTinterwordspacing
A.~Garcia{-}Saavedra, A.~Banchs, P.~Serrano, and J.~Widmer, ``Distributed
  opportunistic scheduling: {A} control theoretic approach,'' in
  \emph{Proceedings of the {IEEE} {INFOCOM} 2012, Orlando, FL, USA, March
  25-30, 2012}, 2012, pp. 540--548. [Online]. Available:
  \url{https://doi.org/10.1109/INFCOM.2012.6195795}
\BIBentrySTDinterwordspacing

\bibitem{SerranoPMMB13}
\BIBentryALTinterwordspacing
P.~Serrano, P.~Patras, A.~Mannocci, V.~Mancuso, and A.~Banchs, ``Control
  theoretic optimization of 802.11 wlans: Implementation and experimental
  evaluation,'' \emph{Computer Networks}, vol.~57, no.~1, pp. 258--272, 2013.
  [Online]. Available: \url{https://doi.org/10.1016/j.comnet.2012.09.010}
\BIBentrySTDinterwordspacing

\bibitem{Kelly1998}
\BIBentryALTinterwordspacing
F.~P. Kelly, A.~K. Maulloo, and D.~K.~H. Tan, ``Rate control for communication
  networks: shadow prices, proportional fairness and stability,'' \emph{Journal
  of the Operational Research Society}, vol.~49, no.~3, pp. 237--252, Mar 1998.
  [Online]. Available: \url{https://doi.org/10.1057/palgrave.jors.2600523}
\BIBentrySTDinterwordspacing

\bibitem{JinWL04}
\BIBentryALTinterwordspacing
C.~Jin, D.~X. Wei, and S.~H. Low, ``{FAST} {TCP:} motivation, architecture,
  algorithms, and performance,'' in \emph{Proceedings {IEEE} {INFOCOM} 2004,
  The 23rd Annual Joint Conference of the {IEEE} Computer and Communications
  Societies, Hong Kong, China, March 7-11, 2004}, 2004, pp. 2490--2501.
  [Online]. Available: \url{https://doi.org/10.1109/INFCOM.2004.1354670}
\BIBentrySTDinterwordspacing

\bibitem{BBR}
N.~Cardwell, Y.~Cheng, C.~S. Gunn, S.~H. Yeganeh, and V.~Jacobson, ``Bbr:
  Congestion-based congestion control,'' \emph{Commun. ACM}, vol.~60, no.~2,
  pp. 58--66, 2017.

\bibitem{model}
F.~Gringoli and D.~J. Leith, ``{{Modelling Downlink Packet Aggregation in Paced
  802.11ac WLANs}},'' 2021, tech Report, Jan 2021.

\bibitem{leith00}
D.~Leith and W.~Leithead, ``Survey of gain-scheduling analysis and design,''
  \emph{Int. J.Control}, vol.~73, no.~11, pp. 1001--1025, 2000.

\bibitem{Cubic}
S.~Ha, I.~Rhee, and L.~Xu, ``Cubic: A new tcp-friendly high-speed tcp
  variant,'' \emph{SIGOPS Oper. Syst. Rev.}, vol.~42, no.~5, pp. 64--74, 2008.

\bibitem{Parisa}
Y.~Im, P.~Rahimzadeh, B.~Shouse, S.~Park, C.~Joe-Wong, K.~Lee, and S.~Ha, ``I
  sent it: Where does slow data go to wait?'' in \emph{Proc Fourteenth EuroSys
  Conf}, ser. EuroSys '19.\hskip 1em plus 0.5em minus 0.4em\relax New York, NY,
  USA: ACM, 2019, pp. 22:1--22:15.

\end{thebibliography}

%
\end{document}